\documentclass[sigconf]{acmart}
\pdfoutput=1
\renewcommand\footnotetextcopyrightpermission[1]{} 
\usepackage{natbib}
\usepackage{mathtools}
\usepackage{amsthm}

\newtheorem*{problem}{Problem Statement}

\newtheorem*{theorem*}{Theorem}

\begin{document}
\sloppy
\title{Timely-Throughput Optimal Coded Computing over Cloud Networks}

\author{Chien-Sheng Yang}
\affiliation{%
  \institution{University of Southern California}
}
\email{chienshy@usc.edu}

\author{Ramtin Pedarsani}
\affiliation{%
  \institution{University of California, Santa Barbara}
}
\email{ramtin@ece.ucsb.edu}

\author{A. Salman Avestimehr}
\affiliation{
  \institution{University of Southern California}
}
\email{avestimehr@ee.usc.edu}


\begin{abstract}
 In modern distributed computing systems, unpredictable and unreliable infrastructures result in high variability of computing resources. Meanwhile, there is significantly increasing demand for timely and event-driven services with deadline constraints.  
 %
 Motivated by measurements over Amazon EC2 clusters, we consider a two-state Markov model for variability of computing speed in cloud networks. In this model, each worker can be either in a good state or a bad state in terms of the computation speed, and the transition between these states is modeled as a Markov chain which is unknown to the scheduler. 
 We then consider a \emph{Coded Computing} framework, in which the data is possibly encoded and stored at the worker nodes in order to provide robustness against nodes that may be in a bad state.
 With timely computation requests submitted to the system with computation deadlines, our goal is to design the optimal computation-load allocation scheme and the optimal data encoding scheme that maximize the timely computation throughput (i.e, the average number of computation tasks that are accomplished before their deadline). 
 Our main result is the development of a dynamic computation strategy called \emph{Lagrange Estimate-and-Allocate (LEA)} strategy, which achieves the optimal timely computation throughput. 
 %
 %
 It is shown that compared to the static allocation strategy, LEA increases the timely computation throughput by
 $1.4 \times \sim 17.5 \times$ in various scenarios via simulations and by $1.27 \times \sim 6.5 \times$ in experiments over Amazon EC2 clusters.

\end{abstract}

%
%

\maketitle
\section{Introduction}\label{sec:intro}

Large-scale distributed computing systems can substantially suffer from unpredictable and unreliable computing infrastructure which can result in high variability of computing resources, i.e., speed of the computing resources vary over time. The speed variation has several causes including hardware failure, co-location of computation tasks, communication bottlenecks, etc. \cite{zaharia2008improving,ananthanarayanan2013effective} This variability is further amplified in computing clusters, such as Amazon EC2, due to the utilization of credit-based computing policy, in which the most commonly used T2 and T3 instances can operate significantly above a baseline level of CPU performance (approximately $10$ times faster as shown in Fig. \ref{fig:credit_speed}) by consuming CPU credits that are allocated periodically to the nodes. At the same time, there is a significant increase in utilizing the cloud  for event-driven and time-sensitive computations (e.g., IoT applications and cognitive services), in which the users increasingly demand timely services with deadline constraints, i.e., computations of requests have to be finished within specified deadlines.

%
Our goal in this paper is to study the problem of computation allocation over cloud networks with particular focus on variability of computing resources and timely computation tasks. From the measurements of nodes' computation speeds over Amazon EC2 clusters, shown in Fig. \ref{fig:credit_speed}, we observe that when a node is slow (fast), it is more likely that it continues to be slow (fast) in the following rounds of computation, which implies temporal correlation of computation speeds. Thus, to capture this phenomenon, we consider a two-state Markov model for variability of computing speed in cloud networks. In this model, each worker can be either in a good state or a bad state in terms of the computation speed, and the transition between these states is modeled as a Markov chain which is unknown to the scheduler.

Furthermore, we consider a \emph{Coded Computing} framework, in which the data is possibly encoded and stored at the worker nodes in order to provide robustness against nodes that may be in a bad state. The key idea of coded computing is to encode the data and design each worker's computation task such that the fastest responses of any $k$ workers out of total of $n$ workers suffice to complete the distributed computation, similar to classical coding theory where receiving any $k$ symbols out of $n$ transmitted symbols enables the receiver to decode the sent message.

We consider a dynamic computation model, where a sequence of functions needs to be computed over the (encoded) data that is distributedly stored at the nodes. More precisely, in an online manner, timely computation requests with given deadlines are submitted to the system, i.e., each computation has to be finished within the given deadline. Our goal is then to design the optimal computation-load allocation strategy and the optimal data encoding scheme that maximize the timely computation throughput (i.e, the average number of computation tasks that are accomplished before their deadline).\footnote{Our metric of timely computation throughput is motivated by timely throughput metric, introduced in~\cite{timelyThroughput}, which measures the average number of packets that are delivered by their deadline in a communication network.}

One significant challenge in this problem is the joint design of (1) a data encoding scheme to provide robustness against straggling workers; and (2) an adaptive computation load allocation strategy for the workers based on the history of previous computation times. In particular, due to the fact that the state of the computing nodes and the transition probabilities of the Markov model are unknown to the scheduler. We note that to find the optimal computation strategy, one has to solve a complex optimization which in general requires searching over all possible load allocations, even if the transition probabilities of Markov model are known to the master. Thus, it is not clear how one allocates the computation loads efficiently and what computation strategy is optimal, especially for the network with unknown Markov model.
 \begin{figure}[t]
  \centering
    \includegraphics[width = 0.9\columnwidth]{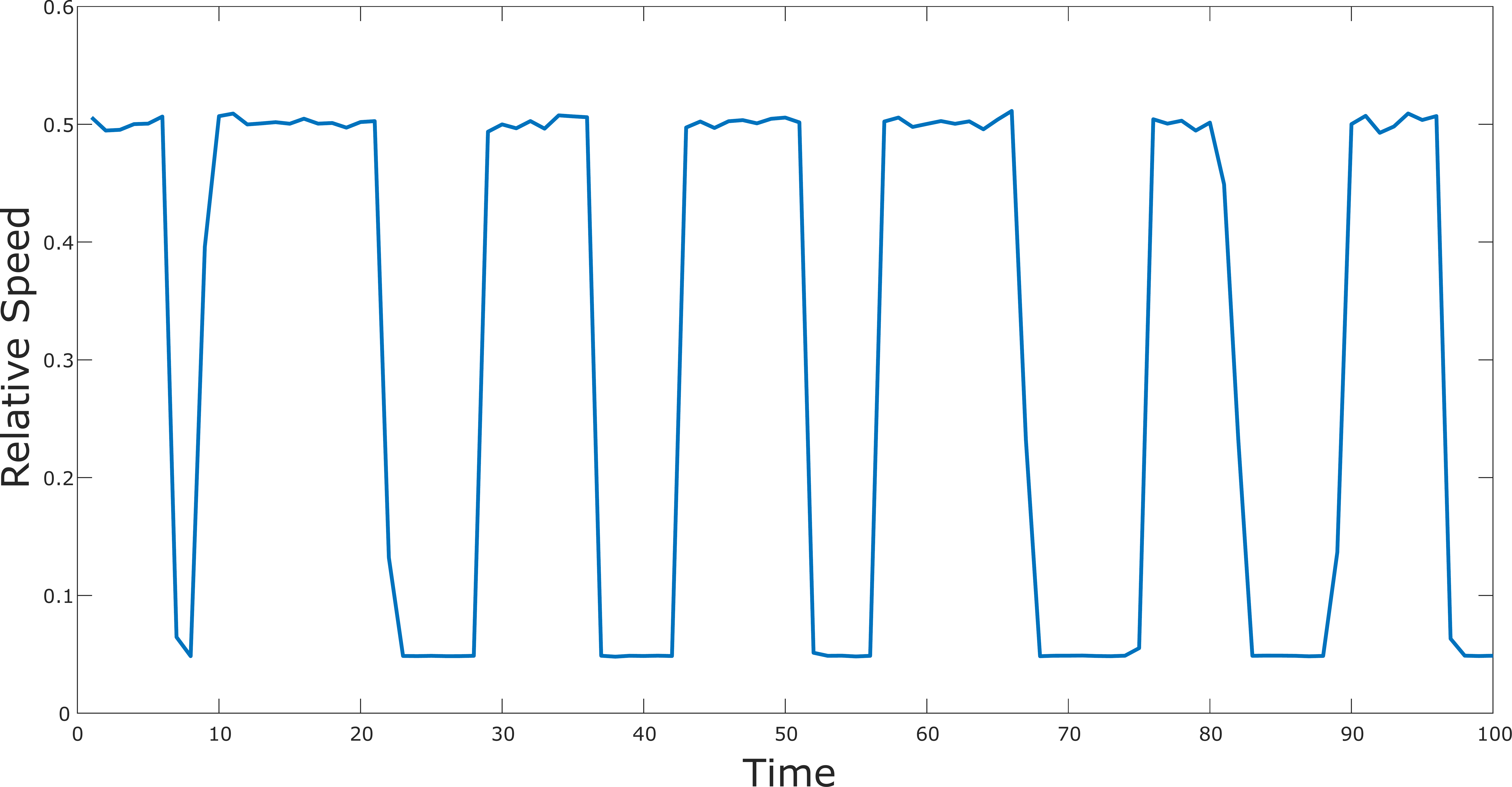}
\caption{Empirical measurement of speed variation of a credit-based \texttt{t2.micro} instance in Amazon EC2 in which we keep assigning computation (e.g., a matrix multiplication) to the instance and measure the finish times: A two-state Markov model.}
\label{fig:credit_speed}
\end{figure}

As the main contributions of the paper, we propose a dynamic computation strategy called \emph{Lagrange Estimate-and-Allocate (LEA)} strategy, and show that it achieves the optimal timely computation throughput. Utilizing Lagrange coding scheme for data encoding \cite{yu2019lagrange}, the LEA strategy estimates the transition probabilities by observing the past events at each time step, and then assigns computation loads based on the estimated probabilities. Moreover, we also show that finding the optimal load assignment using LEA can be done efficiently instead of searching over all possible load allocations which is computationally infeasible to implement.   

To prove the optimality of LEA strategy, we first focus on finding the optimal timely-throughput by maximizing the success probability of each round when the transition probabilities are known to master. For any fixed load assignment, we show that using Lagrange coding scheme proposed in \cite{yu2019lagrange} has the highest success probability of each round. Then, we show that the success probability using LEA converges to the optimal success probability. By the Strong Law of Large Numbers (SLLN), Ergodic theorem and a coupling argument, we finally prove that timely computation throughput achieved by the LEA strategy is equal to the optimal timely computation throughput, i.e., LEA is optimal. 

 In addition to proving the optimality of LEA, we carry out numerical studies and experiments over Amazon EC2 clusters. We compare the proposed LEA strategy with a static load allocation strategy for the benchmark. In our numerical analysis, compared to the static computation strategy, the LEA strategy increases the timely computation throughput by $1.38 \times \sim 17.5 \times$. In experiments over Amazon EC2 clusters, the LEA strategy increases the timely computation throughput by $1.27 \times \sim 6.5 \times$.
\subsection{Related Prior Work}
We divide the literature review to two main lines of work: scheduling and load balancing over cloud networks, and coded computing in distributed systems. 
 
\textbf{Task Scheduling:} Task scheduling problem has been widely studied in the literature, which can be divided into two main categories: static scheduling and dynamic scheduling. In the static or offline scheduling problem, jobs are present at the beginning, and the goal is to allocate tasks to servers such that a performance metric such as average computation delay is minimized. In most cases, the static scheduling problem is computationally hard, and various heuristics, approximation and stochastic approaches are proposed (see e.g. \cite{kwok1999static,zheng2013stochastic,topcuoglu2002performance}). 

In the dynamic or online scheduling problem, jobs arrive to the network according to a stochastic process, and get scheduled dynamically over time. In many works in the literature, the tasks have dedicated servers for processing, and the goal is to establish stability conditions for the network \cite{baccelli1989acyclic}. Given the stability results, the next natural goal is to compute the expected completion times of jobs or delay distributions. However, few analytical results are available for characterizing the delay performance, except for the simplest models. When the tasks do not have dedicated servers, one aims to find a throughput-optimal scheduling policy (see e.g. \cite{eryilmaz2005stable}), i.e. a policy that stabilizes the network, whenever it can be stabilized. For example, Max-Weight scheduling, proposed in \cite{tassiulas1992stability,dai2005maximum}, is known to be throughput-optimal for wireless networks, flexible queueing networks \cite{neely2005dynamic,eryilmaz2007fair,pedarsani2017robust}, data centers networks \cite{maguluri2012stochastic} and dispersed computing networks \cite{Yang1806:Communication}. Moreover, there have been many works which focus on task scheduling problem with deadline constraints over cloud networks (see e.g. \cite{arabnejad2017scheduling,hoseinnejhad2017deadline}). 

\textbf{Coded Computing:} Coded computing is a recently developed area that proposes to inject clever redundancy in the form of ``coded'' data to tackle two major bottlenecks in distributed computing: straggling servers and communication bandwidth \cite{lee2018speeding,li2018fundamental}.  There have been many works that following this line including those that alleviate stragglers (e.g., \cite{dutta2016short,tandon2017gradient,reisizadeh2017coded,yu2017polynomial}), and those that tackle communication bandwidth (e.g., \cite{li2017coding,prakash2018coded}). More recently coded computing has also been utilized to address security and privacy challenges in distributed computing (e.g., \cite{chen2018draco,bitar2017minimizing,yu2019lagrange}). 

 So far, research in coded computing has focused on developing frameworks for one round of computation instead of considering network dynamics for analyzing long-run performance of distributed computing systems. In this paper, considering the dynamics of the network, we make substantial progress by combining the ideas of coded computing with dynamic computation load allocation over cloud networks, and developing Lagrange Estimate-and Allocate strategy that can adaptively assign computation loads to workers and essentially learn the unknown network dynamics. Furthermore, we consider the metric "timely computation throughput" which denotes the average number of successful completions instead of the metric "timely throughput" which usually denotes the average number of packets delivered successfully in network scenarios (see e.g., \cite{lashgari2013timely}).  

\section{System Model}\label{sec:sys}
\begin{figure}[t]
    \centering
    \includegraphics[width =\columnwidth]{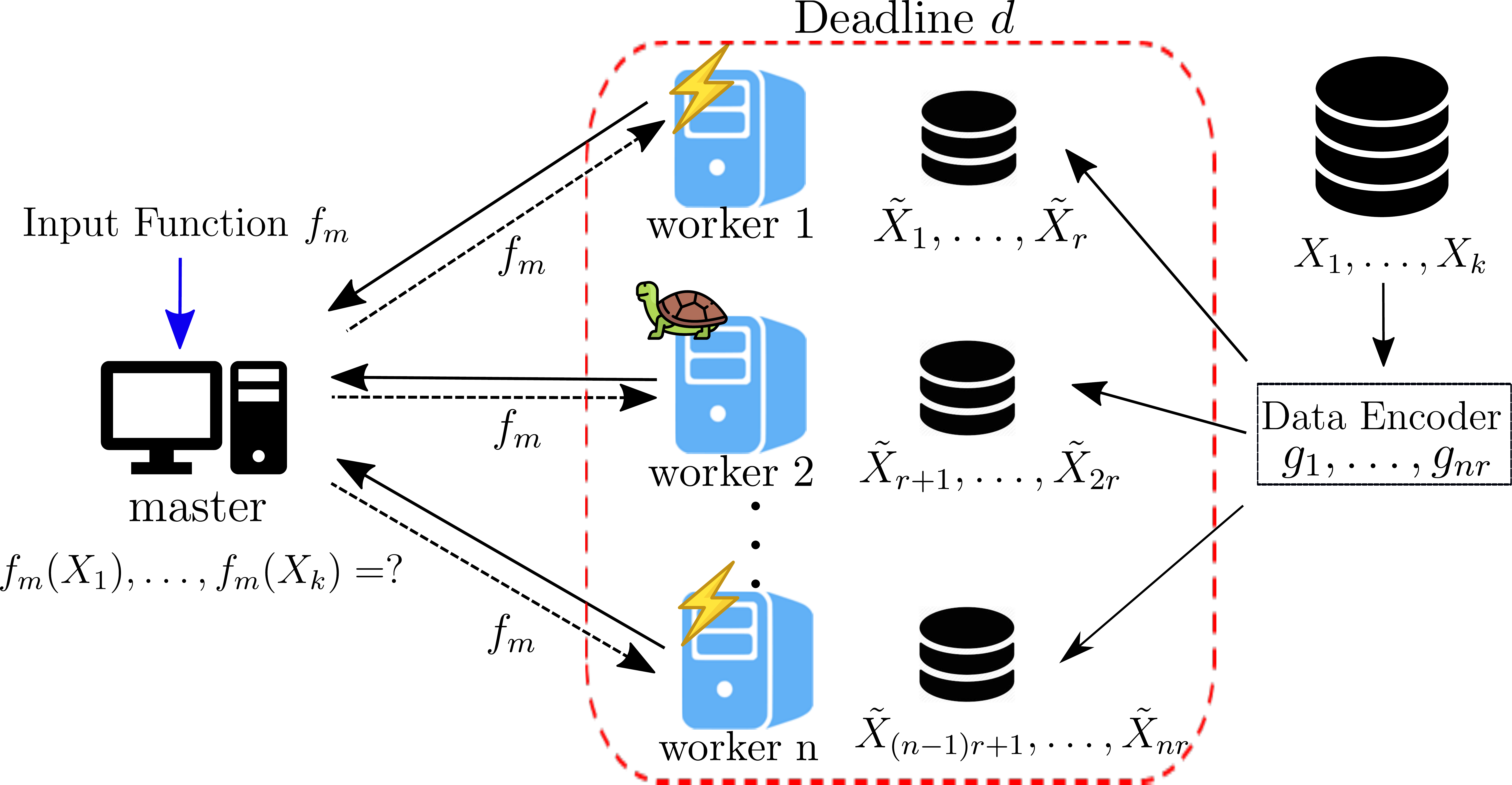}
    \caption{Overview of dynamic load allocation over a coded computing framework with timely computation requests. In each round $m$, the goal is to compute the evaluations $f(X_1),\dots,f(X_k)$ by the deadline $d$ using $n$ workers.
    } 
    \label{fig:system}
\end{figure}
\subsection{Computation Model}\label{label:comp_model}
We consider a distributed computing problem, in which computation requests are submitted to a distributed computing system in an online manner, and the computation is carried out in the system. In particular, there is a fixed deadline for each computation round, i.e., each computation has to be finished within the given deadline.

As shown in Fig. \ref{fig:system}, the considered system is composed of a master node and $n$ worker nodes. There is also a dataset $X$ which is divided to $X_1,X_2,\dots,X_k$. Specifically, each $X_j$ is an element in a vector space $\mathbb{V}$ over a field $\mathbb{F}$. In each round $m$ (or time slot in a discrete-time system), a computation request with a function $f_m$ is submitted to the system, where the function $f_m$ is an arbitrary multivariate polynomial with vector coefficients having degree \textrm{deg}$(f)$. We denote by $d$ the deadline of each computation request which is smaller than or equal to the duration of each round. In such distributed computing system, we are interested in computing the evaluations $f_m(X_1),f_m(X_2),\dots,f_m(X_k)$ in each round $m$ by the deadline $d$. 

Prior to the computation, the master first encodes the dataset $X_1,X_2,\dots,X_k$ to $\tilde{X}_1$, $\tilde{X}_2$, $\dots$, $\tilde{X}_{nr}$ via a set of $nr$ encoding functions $\vec{g} = (g_1,g_2,$ $\dots,g_{nr})$, where encoded data $\tilde{X}_v \triangleq g_v(X_1,\dots,X_k)$ is determined by the encoding function $g_v:\mathbb{V} \rightarrow \mathbb{U}$. Each worker $i$ stores $r$ encoded data chunks $\tilde{X}_{(i-1)r+1}$, $\tilde{X}_{(i-1)r+2}$ ,$\dots,\tilde{X}_{ir}$ locally. In each round $m$, each worker evaluates certain subset of $f_m(\tilde{X}_{(i-1)r+1})$, $f_m(\tilde{X}_{(i-1)r+2})$, $\dots$, $f_m(\tilde{X}_{ir})$ determined by the master.

Given a function $f_m$ in round $m$, the master assigns the computations to each worker. More specifically, we define $\vec{\ell}_m=(\ell_{m,1},\ell_{m,2},\dots,\ell_{m,n})$ to be the load allocation vector, in which $\ell_{m,i}$ denotes the number of polynomial or function evaluations computed by worker $i$ in round $m$. Each worker $i$ computes $\ell_{m,i}$ evaluations of function $f_m$ over the stored data without specified order, and returns all the results back to the master upon the completion of all assigned computations. The master node aggregates the results from the worker nodes until it receives a \emph{decodable} set of computations and recovers $f_m(X_1),f_m(X_2),\dots,f_m(X_k)$. We say a set of computations is decodable if the evaluations $f_m(X_1)$, $f_m(X_2)$, $\dots$, $f_m(X_k)$ can be obtained by computing decoding functions over received results. In each round, the goal of the master is to receive a decodable set of computations within the given deadline $d$.

Let us illustrate the model through a simple example.

\textbf{Example.} In each round $m$, we consider a problem of evaluating a linear function $f_m(X_j) = X_j\vec{w}_m$ over $n = 3$ workers, where the input dataset $X$ is divided to $X_1$, $X_2$ and $\vec{w}_m$ is the input vector. One possible coding scheme is to encode $X_1$ and $X_2$ to $\tilde{X}_1 = X_1$, $\tilde{X}_2 = X_2$ and $\tilde{X}_3 = X_1 + X_2$. Each worker $i$ stores $r=1$ encoded data chunk $\tilde{X}_i$. If the load allocation vector $\vec{\ell}_m = (1,1,1)$ is used by the master, then each worker $i$ computes $\tilde{X}_i\vec{w}_m$ and sends the result back the master upon its completion. The set $\{\tilde{X}_1\vec{w}_m, \tilde{X}_3\vec{w}_m\}$ is one of decodable sets since the master can obtain $X_1\vec{w}_m$ and $X_2\vec{w}_m$ by computing $X_1\vec{w}_m = \tilde{X}_1\vec{w}_m$ and $X_2\vec{w}_m = \tilde{X}_3\vec{w}_m-\tilde{X}_1\vec{w}_m$. 

We note that the considered computation model naturally appears in many gradient computing problems. For example, in linear regression problems, we want to compute $f_m(X_j)$ $=X_j^{\top}(X_j\vec{w}_m-\vec{y})$ which is the gradient of the quadratic loss function $\frac{1}{2}(X_j^{\top}\vec{w}_m-\vec{y})^2$ with respect to the weight vector $\vec{w}_m$ in round $m$. 

\subsection{Network Model}\label{label:net_model}
Motivated by the measurements over Amazon EC2 clusters, shown in Fig. \ref{fig:credit_speed}, we assume that each worker has two different states for computing, \textit{good state} and \textit{bad state}. We denote $\mu_g$ as the computing speed (evaluations per second) in the good state, and denote $\mu_b$ as the computing speed in the bad state. We assume that the computing speeds $\mu_g$ and $\mu_b$ are known to the master. Note that given a worker's state, its computation time (per evaluation) is deterministic. We denote $\mu_{m,i}$ as computing speed of worker $i$ in round $m$. And, we denote $\vec{\mu}_m= (\mu_{m,1},\mu_{m,2},\dots,\mu_{m,n})$ as computing speed vector in round $m$. For each worker $i$, we model the state transitions as a stationary Markov process $S_{i}[1],S_{i}[2],\dots,$. The transition matrix for worker $i$ is defined as follows:
\begin{align}
    P_i =
    \begin{bmatrix}
    p_{g \rightarrow g,i} & 1 - p_{g \rightarrow g,i}\\
    1 - p_{b \rightarrow b,i} & p_{b \rightarrow b,i}
    \end{bmatrix}
\end{align}
where $p_{g \rightarrow g,i}$ is the transition probability of worker $i$ going to the good state from the good state, and $p_{b \rightarrow b,i}$ is the transition probability of worker $i$ going to the bad state from the bad state. We assume that the Markov processes of different workers are mutually independent. Prior to the computation, we assume the initial state of worker $i$ is given by the stationary distribution of Markov chain ($S_i[1],S_i[2],\dots$). We assume that the transition probabilities and current state of each worker are unknown to the master before the master assigns the computations to each worker. 
\subsection{Problem Formulation}
Given the computation deadline $d$, we denote $N_m(d)$ as an indicator representing whether the computation is finished by deadline $d$, i.e., $N_m(d)=1$ if the computation is finished by time $d$ in round $m$, and $N_m(d)=0$ otherwise. We denote $\eta =(\vec{g},\{\vec{\ell}_m\}^{\infty}_{m=1})$ as the computation strategy. Also, we denote the set of all computation strategies as $\Gamma$.
 \begin{definition}[Timely Computation Throughput]
Given the computation deadline $d$, using computation strategy $\eta$, the timely computation throughput, denoted by $R(d,\eta)$, is defined as follows: 
\begin{align}
R(d,\eta) = \lim_{M \rightarrow \infty} \frac{\sum^M_{m=1}N_m(d)}{M}.
\end{align}
\end{definition}
Based on the above definitions, our problem is now formulated as the following.
\begin{problem}
Consider a distributed computing system consisting
of computation and network models as defined in Subsections \ref{label:comp_model} and \ref{label:net_model}. Our goal is to find an optimal computation strategy achieving optimal timely computation throughput, denoted by $R^{*}(d)$ which is defined as follows:
\begin{align}
    R^{*}(d) = \sup_{\eta \in \Gamma}R(d,\eta)
\end{align}
\end{problem}

\section{Lagrange Estimate-and-Allocate (LEA) Strategy}\label{sec:EA}
In this section, we propose a dynamic computation strategy called \emph{Lagrange Estimate-and-Allocate (LEA)} strategy, which is composed of Lagrange coding scheme for data encoding and \emph{Estimate-and-Allocate (EA)} algorithm for allocating loads to the workers adaptively by observing the history of computation times. In each round, the EA algorithm first assigns computation loads by maximizing the estimated success probability based on the estimated transition probabilities of the underlying Markov chain (and based on that the previous state of the workers). After receiving the results, the EA algorithm updates the estimated transition probabilities by observing the computation times in the past events.
\subsection{Data Encoding in LEA} \label{subsec:LCC}
For data encoding, we leverage a linear coding scheme called Lagrange coding scheme which is proposed in \cite{yu2019lagrange}. We start with an illustrative example.

We first consider the scenario where $nr \geq k \ \textrm{deg} (f)-1$. In each round $m$, we consider a problem of evaluating a quadratic function $f_m(X_j)$ $=X_j^{\top}X_j\vec{w}_m$ (\textrm{deg}$(f)$=2) over $n=3$ workers, where the input dataset $X$ is divided to $X_1,X_2$. Each worker stores $r=2$ encoded data chunks ($nr=6 > k \ \textrm{deg} (f)-1 =3$). We define $u$ as follows:
\begin{align}
    u(z) \triangleq X_1\frac{z-1}{0-1}+X_2 \frac{z-0}{1-0} = z(X_2-X_1)+X_1,
\end{align}
in which $u(0) = X_1$ and $u(1)=X_2$. Then, we encode $X_1$ and $X_2$ to $\tilde{X}_i = u(i-1)$, i.e., $\tilde{X}_1=X_1$, $\tilde{X}_2=X_2$, $\tilde{X}_3=-X_1+2X_2$, $\tilde{X}_4=-2X_1+3X_2$, $\tilde{X}_5=-3X_1+4X_2$ and $\tilde{X}_6=-4X_1+5X_2$. Each worker $i$ stores $\tilde{X}_{2i-1}$ and $\tilde{X}_{2i}$ locally. 

We now consider the scenario where $nr < k \ \textrm{deg} (f)-1$. We consider the same problem in the previous scenario, but the there is larger input dataset $X$ which is divided to $X_1,X_2,X_3$ and $X_4$ ($nr=6 < k \ \textrm{deg} (f)-1 =7$). We encode $X_1$ and $X_2$ using a repetition coding design such that $\tilde{X}_1=X_1$, $\tilde{X}_2=X_2$, $\tilde{X}_3=X_3$, $\tilde{X}_4=X_4$, $\tilde{X}_5=X_1$ and $\tilde{X}_6=X_2$. Each worker $i$ stores $\tilde{X}_{2i-1}$ and $\tilde{X}_{2i}$ locally. 

Formally, we describe Lagrange coding scheme as follows:\\
(1) $nr \geq k \ \textrm{deg} (f)-1$: We first select $k$ distinct elements $\beta_1, \beta_2,\dots, \beta_k$ from $\mathbb{F}$, and let $u$ be the respective \emph{Lagrange interpolation polynomial} 
\begin{align}
    u(z) \triangleq \sum^{k}_{j=1}X_j \prod_{l \in [k]\backslash \{j\}}\frac{z-\beta_l}{\beta_j - \beta_l}.
\end{align}
where $u: \mathbb{F} \rightarrow \mathbb{V}$ is a polynomial of degree $k-1$ such that $u(\beta_j) = X_j$. To encode the input $X_1,X_2,\dots,X_k$, we select $nr$ distinct elements $\alpha_1,\alpha_2,\dots,\alpha_{nr}$ from $\mathbb{F}$, and encode $X_1,X_2,\dots,X_k$ to $\tilde{X}_v=u(\alpha_v)$ for all $v \in [nr]$, i.e.,
\begin{align}
    \tilde{X}_v = g_v(X) = u(\alpha_v) \triangleq  \sum^{k}_{j=1}X_j \prod_{l \in [k]\backslash \{j\}}\frac{\alpha_v-\beta_l}{\beta_j - \beta_l}.
\end{align}
Each worker $i$ stores $\tilde{X}_{(i-1)r+1},\tilde{X}_{(i-1)r+2},\dots,\tilde{X}_{ir}$ locally. \\
(2) $nr < k \ \textrm{deg} (f)-1$: We use a repetition coding design to encode the input $X_1,X_2,\dots,X_k$. We replicate every $X_i$ either $\lfloor \frac{nr}{k} \rfloor$ or $\lceil \frac{nr}{k} \rceil$ times such that the number of total encoded data chunks is $nr$. Then, we obtain the encoded data $\tilde{X}_{1},\tilde{X}_{2},\dots,\tilde{X}_{nr}$. Each worker picks $r$ of the encoded data $\tilde{X}_{1},\tilde{X}_{2},\dots,\tilde{X}_{nr}$ to be stored locally.

Note that decoding and encoding in Lagrange coding scheme relies on polynomial interpolation and evaluation which can be done efficiently.
\subsection{Load Allocation in LEA}
Before introducing the EA algorithm, we first define the following terms. For each worker $i$, we denote $C_{g\rightarrow g,i}(m)$ as the number of times that event "good state to good state" happened up to round $m$, $C_{g\rightarrow b,i}(m)$ as the number of times that event "good state to bad state" happened up to round $m$, $C_{b\rightarrow g,i}(m)$ as the number of times that event "bad state to good state" happened up to round $m$ and $C_{b\rightarrow b,i}(m)$ as the number of times that event "bad state to bad state" happened up to round $m$. 

For worker $i$, we denote $\hat{p}_{g \rightarrow g,i}(m)$ and $\hat{p}_{b \rightarrow b,i}(m)$ as the estimated transition probabilities after the first $m-1$ rounds of computations. For worker $i$, we denote $\hat{p}_{g,i}(m)$ and $\hat{p}_{b,i}(m)$ as the estimated probabilities being in the good state and the bad state in round $m$ respectively. Without loss of generality, we assume that $\hat{p}_{g,1}(m) \geq \hat{p}_{g,2}(m) \geq \dots \geq \hat{p}_{g,n}(m)$. We also define $\ell_b \triangleq \mu_bd$ and $\ell_g \triangleq \min(\mu_gd,r)$

Now, we formally describe the EA algorithm. In each round $m$, the EA algorithm has the following $4$ phases:

\textbf{(1) Load Assignment Phase}:
The master maximizes the estimated success probability in round $m$ based on the the estimated probabilities $\hat{p}_{g,i}(m)$ and $\hat{p}_{b,i}(m)$. To do so, the master finds $i^*_m$ $(1 \leq i^*_m \leq n)$ maximizing the estimated success probability function defined as follows\footnote{Note that we only consider the case: $K^* \geq n\mu_bd =n\ell_b$, otherwise the computation can be always finished in time $d$ which is trivial.}:
\begin{align} \label{eq:est_success1}
    \hat{\mathbb{P}}_m(\tilde{i}) =  
    0  \  \text{if} \ K^* > \tilde{i}\ell_g +(n-\tilde{i})\ell_b,
\end{align}
otherwise
\begin{align} \label{eq:est_success2}
  \hat{\mathbb{P}}_m(\tilde{i}) =    \sum^{\tilde{i}}_{l=w(\tilde{i})}\sum_{\mathcal{G}:\mathcal{G}\subseteq [\tilde{i}],|\mathcal{G}|=l} \prod_{i \in \mathcal{G}} \hat{p}_{g,i}(m) \prod_{i \in [\tilde{i}] \backslash \mathcal{G}} \hat{p}_{b,i}(m)
\end{align}
where $w(\tilde{i}) \triangleq \lceil \frac{K^* - (n-\tilde{i})\ell_b}{\ell_g} \rceil$ and $K^*$ is defined as follows:
\begin{align}\label{eq:recover_threshold}
    K^{*} = 
    \begin{cases}
        (k-1) \textrm{deg}(f) + 1  \ & \text{if} \ nr \geq k \ \textrm{deg} (f)-1 \\
        nr-\lfloor\frac{nr}{k}\rfloor+1 \ & \text{otherwise}.
    \end{cases}
\end{align}
Note that equations (\ref{eq:est_success1}) and (\ref{eq:est_success2}) define the estimated success probability which is the function of $\tilde{i}$ (number of workers assigned to compute $\ell_g$ evaluations). The intuition behind equation (\ref{eq:est_success1}) is that if total load assigned to all the workers is smaller than the optimal recovery threshold, the probability of success is zero. Based on the estimated probabilities $\hat{p}_g$ and $\hat{p}_b$, equation (\ref{eq:est_success2}) gives us the estimated success probability by summing the probabilities of events which have enough workers in good state leading to successful completion of the computation before the deadline. Also, $K^*$ defined in (\ref{eq:recover_threshold}) is the optimal recovery threshold using Lagrange coding scheme \cite{yu2019lagrange} which guarantees that the evaluations can be recovered when the master receives any $K^*$ results from the workers. Thus, $i^*_m = \arg\max \hat{\mathbb{P}}_m(\tilde{i})$. Then, the master does assignment by using the load allocation vector $\ell_{m}$ such that
\begin{align}
\ell_{m,i} =
\begin{cases} 
    \ell_g, \ \text{if} \ 1 \leq i \leq i^*_m\\
    \ell_b, \ \text{otherwise}.
\end{cases}
\end{align}
In load assignment phase, the idea is to select workers in the order of the estimated probability being in the good state, and assign more loads accordingly. Note that it is just a linear search in load assignment phase which is computationally efficient.

\textbf{(2) Local Computation Phase:} Within each round $m$ of computation, each worker $i$ receives function $f_m$ and load assignment $\ell_{m,i}$ from the master. Then, each worker $i$ computes evaluations of function $f_m$ over encoded data $\tilde{X}_{(i-1)r+1},\tilde{X}_{(i-1)r+2},\dots,\tilde{X}_{(i-1)r+\ell_{m,i}}$, i.e., $f_m(\tilde{X}_{(i-1)r+1}),f_m(\tilde{X}_{(i-1)r+2}),\dots,f_m(\tilde{X}_{(i-1)r+\ell_{m,i}})$. After the computation, each worker sends all the computation results back to the master upon its completion. 

\textbf{(3) Aggregation and Observation Phase}: Having received the fastest $K^*$ computation results from the workers, the master recovers the evaluations $f_m(X_1), f_m(X_2), \dots, f_m(X_k)$ for the request function $f_m$. By observing whether the results are sent back or not, the master checks which one of events "good state to good state", "good state to bad state", "bad state to good state" and "bad state to bad state" has happened in round $m$ for each worker $i$. Then, the master obtains $C_{g \rightarrow g,i}(m)$, $C_{g \rightarrow b,i}(m)$, $C_{b \rightarrow g,i}(m)$ and $C_{b \rightarrow b,i}(m)$. Note that the time that it takes for one worker's result to be completed and sent back to the master actually indicates the (previous) state of that worker, since the speeds are deterministic and the computation time in a good state is less than the computation time in a bad state.

\textbf{(4) Update Phase}: After aggregation and observation phase, the master updates the estimated transition probabilities $\hat{p}_{g \rightarrow g,i}(m+1)$ and $\hat{p}_{b \rightarrow b,i}(m+1)$ for the round $m+1$: $\hat{p}_{g \rightarrow g,i}(m+1) = \frac{C_{g\rightarrow g,i}(m)}{C_{g\rightarrow g,i}(m)+ C_{g\rightarrow b,i}(m)}$ and $\hat{p}_{b \rightarrow b,i}(m+1) = \frac{C_{b\rightarrow b,i}(m)}{C_{b\rightarrow g,i}(m)+ C_{b\rightarrow b,i}(m)}$.
The master updates the estimated probabilities $\hat{p}_{g,i}(m+1)$ and $\hat{p}_{b,i}(m+1)$. If worker $i$ was in good state in round $m$, $\hat{p}_{g,i}(m+1) = \hat{p}_{g \rightarrow g,i}(m+1)$, and  $ \hat{p}_{g,i}(m+1) =1- \hat{p}_{b \rightarrow b,i}(m+1)$ otherwise. Then, the computation goes to the round $m+1$.
\section{Upper bound on the timely computation throughput} \label{sec:timelythroughput}
In this section, we give an upper bound for the timely computation throughput. The idea is to consider the case that the Markov model of the network is known to the master and achieve the optimal computation throughput for this case.



\subsection{Optimal Success Probability of One Round Computation}\label{subsec:success_prob}
First, we consider one round of computation using a load allocation vector $\vec{\ell}$ with a linear coding scheme $\vec{g}$. Without knowing computing speed vector $\vec{\mu}$, we denote $T^{(\vec{\ell},\vec{g})}(\vec{\mu})$ as the random variable of finish time using $\vec{\ell}$ and $\vec{g}$. We define the success probability as the probability that the computation is finished in time $d$, i.e., $\mathbb{P}(T^{(\vec{\ell},\vec{g})} \leq d)$ according to the distribution of $\vec{\mu}$. 

For a coding scheme, we define \emph{recovery threshold} which is formally stated as follows:
 \begin{definition}[Recovery Threshold]
 For an integer $k$, a coding scheme $\vec{g}$ is \emph{k-recoverable} if the master can recover the required function evaluations from any $k$ of $nr$ local computation results. We define the \emph{recovery threshold} of a coding scheme $\vec{g}$, denoted by $K(\vec{g})$, as the minimum number of $k$ such that the coding scheme $\vec{g}$ is $k$-recoverable.
 \end{definition}
Given a coding scheme $\vec{g}$, we have the recovery threshold $K(\vec{g})$ which is the minimum number of evaluations to be received in total from the workers. Thus, we aim at finding a coding scheme and a load allocation vector that maximizes the success probability by solving the following optimization problem:
\begin{align}
    \text{Maximize} \ &\mathbb{P}(T^{(\vec{\ell},\vec{g})} \leq d) \\
    \text{subject to} \ & \sum^n_{i=1} \ell_{i} \geq K(\vec{g}), \\
    & 0 \leq \ell_{i} \leq r, \ \ell_i \in \mathbb{Z}, \forall 1 \leq i \leq n.
\end{align}
In the following, we show that the Lagrange coding scheme achieves the highest success probability for any fixed load allocation vector.
Before proving the optimality of Lagrange coding scheme in terms of success probability, we first define \emph{optimal recovery threshold} as follows:
\begin{definition}
 We define the optimal recovery threshold, denoted by $K^*$, as the minimum achievable recovery threshold. Specifically, \begin{align}
     K^* \triangleq \min_{\vec{g}} K(\vec{g}) .
 \end{align}
 \end{definition}
By \cite{yu2019lagrange}, Lagrange coding scheme achieves optimal recovery threshold of evaluating a multivariate polynomial function $f$ (total degree \textrm{deg}$(f)$) on a dataset of $k$ inputs, which is given by
\begin{align}
    K^{*} = (k-1) \textrm{deg}(f) + 1 \label{eq:recovery1}
\end{align}
when $nr \geq k \ \textrm{deg} (f)-1$, and 
\begin{align}
    K^{*} = nr-\lfloor\frac{nr}{k}\rfloor+1 \label{eq:recovery2}
\end{align}
otherwise. 

 We now show that Lagrange coding scheme achieves the highest success probability for any fixed load allocation vector. It is intuitive that a coding scheme achieving  smaller recovery threshold should have higher success probability. We formally state this claim in the following lemma.
\begin{lemma}(Monotonicity) 
\label{lemma:mon}
Consider an arbitrary load allocation vector $\vec{\ell}$, for any coding schemes $\vec{g_1}$ and $\vec{g_2}$, such that $K(\vec{g}_1) \leq K(\vec{g}_2)$, we have
\begin{align}
    \mathbb{P}(T^{(\vec{\ell},\vec{g}_1)} \leq d) \geq \mathbb{P}(T^{(\vec{\ell},\vec{g}_2)} \leq d).
\end{align}
\end{lemma}
The proof of the lemma \ref{lemma:mon} is provided in the Appendix \ref{appendix:proof_mon}. 

\subsection{Load Allocation Problem} \label{sec:two}
From Lemma \ref{lemma:mon}, by fixing Lagrange coding scheme denoted by $\vec{g^*}$, the optimization problem proposed in Subsection \ref{subsec:success_prob} can be simplified to the optimization problem that only has load allocation vector as variables. We now introduce an optimization problem called \textit{Load Allocation Problem} which is defined as follows:
 
 \textbf{Load Allocation Problem:}
\begin{align}
    \text{Maximize} \ &\mathbb{P}(T^{(\vec{\ell},\vec{g^*})} \leq d) \label{eq:obj2}\\
    \text{subject to} \ & \sum^n_{i=1} \ell_{i} \geq K^*, \\
    & 0 \leq \ell_{i} \leq r, \ \ell_i \in \mathbb{Z}, \forall 1 \leq i \leq n.
\end{align}
where $K^*$ is the optimal recovery threshold defined in (\ref{eq:recovery1}) and (\ref{eq:recovery2}). Note that the proposed load allocation problem is a combinatorial optimization problem that in general requires combinatorial search over all possible allocations to maximize the success probability.

To show that load allocation problem can be solved efficiently, we first present the following lemma whose proof is provided in Appendix \ref{appendix:proof_twovalue}.
\begin{lemma} 
\label{lemma:twovalue}
Given a deadline $d$, if a load allocation vector $\vec{\ell}$ has the success probability $\mathbb{P}(T^{(\vec{\ell},\vec{g^*})}(\vec{\mu}) \leq d)$, then there exists a load allocation vector $\vec{\ell^{'}}$ with success probability $\mathbb{P}(T^{(\vec{\ell^{'}},\vec{g^*})}(\vec{\mu}) \leq d)$ such that $\mathbb{P}(T^{(\vec{\ell^{'}},\vec{g^{*}})}(\vec{\mu}) \leq d) \geq \mathbb{P}(T^{(\vec{\ell},\vec{g^*})}(\vec{\mu}) \leq d)$ and $\ell^{'}_i \in \{\ell_g,\ell_b\}$ where $\ell_g = \min(\mu_gd,r)$ and $\ell_b = \mu_bd$.
\end{lemma}
 By Lemma \ref{lemma:twovalue}, we can focus on finding the optimal load allocation vector by searching all $\vec{\ell}$ satisfying that $\ell_i \in \{\ell_g,\ell_b\}$ for all $i$. To find the optimal load allocation vector, we now consider the load allocation vector characterized by the set $\mathcal{G}_g = \{i:\ell_i = \ell_g\ ,1 \leq i \leq n\}$ which represents the set of workers that computes $\ell_g$ evaluations locally. Once the set $\mathcal{G}_g$ has been determined, $\mathcal{G}_b$ representing the set of workers that computes $\ell_b$ evaluations can be defined as $\{i:i \in [n] \backslash \mathcal{G}_g\}$.
 
 Since $\frac{\ell_b}{\mu_{i}}$ is always less than $d$, the workers in $\mathcal{G}_b$ will always send the results back to the master in time $d$. Since the optimal recovery threshold is $K^*$ using Lagrange coding scheme, the master has to receive at least $K^* -|\mathcal{G}_b|\ell_b$ results from the workers in $\mathcal{G}_g$ to recover the computation in time $d$. That is, there must be at least $\lceil \frac{K^* -|\mathcal{G}_b|\ell_b}{\ell_g} \rceil$ workers in the good state in set $\mathcal{G}_g$. We define $  a(\mathcal{G}_g) \triangleq \lceil \frac{K^*-(n-|\mathcal{G}_g|)\ell_b}{\ell_g} \rceil$ which denotes the minimum number of workers in the good state in set $\mathcal{G}_g$ to guarantee that the master can recover the computation in time $d$. 
  
  Before writing the success probability as a function of $\mathcal{G}_g$, we first define the following terms. We define $T^{(\mathcal{G}_g)}(\vec{\mu})$ as the random variable denoting the finish time using the allocation vector characterized by $\mathcal{G}_g$. We denote $p_{g,i}$ as the probability that worker $i$ is in the good state and $p_{b,i}$ as the probability that worker $i$ is in the bad state. Also, we denote the random variable that represents the number of workers being in good state in set $\mathcal{G}$ as $Q(\mathcal{G})$.
  
 Using the load allocation vector characterized by $\mathcal{G}_g$, we can find the success probability which is a function of $\mathcal{G}_g$ as follows:\\
  (1) $a(\mathcal{G}_g) > |\mathcal{G}_g|$: In this case, the master needs at least $a(\mathcal{G}_g)$ workers being in good state which is greater than $|\mathcal{G}_g|$. It implies that $\mathbb{P}(T^{(\mathcal{G}_g)}(\vec{\mu}) \leq d) =0$.\\
  (2) $0\leq a(\mathcal{G}_g) \leq |\mathcal{G}_g|$: In this case, we have
\begin{align}
     &\mathbb{P}(T^{(\mathcal{G}_g)}(\vec{\mu}) \leq d) =  \mathbb{P}( Q(\mathcal{G}_g) \geq a(\mathcal{G}_g))
    =  \sum^{|\mathcal{G}_g|}_{l=a(\mathcal{G}_g)}\mathbb{P}(Q(\mathcal{G}_g) = l) \nonumber\\
     = &\sum^{|\mathcal{G}_g|}_{l=a(\mathcal{G}_g)}\sum_{\mathcal{G}:\mathcal{G}\subseteq\mathcal{G}_g,|\mathcal{G}|=l} \prod_{i \in \mathcal{G}} p_{g,i} \prod_{i \in \mathcal{G}_g \backslash \mathcal{G}} p_{b,i}. \label{eq:obj_poly}
\end{align}

Therefore, our goal is to find the optimal set $\mathcal{G}^{*}_g$ characterizing the optimal load allocation vector which maximizes the success probability over all possible sets $ \mathcal{G}_g \subseteq [n]$. The complexity of searching over all possible sets $ \mathcal{G}_g \subseteq [n]$ grows exponentially with $n$, since there are overall $2^n$ choices for $\mathcal{G}_g$. 

 The following lemma shows that the optimal $\mathcal{G}^{*}_g$ contains the workers having the largest $p_{g,i}$ among all the workers, which largely reduce the time complexity of finding the optimal $\mathcal{G}^{*}_g$,
\begin{lemma} \label{lemma:optimalset}
Without loss of generality, we assume $p_{g,1} \geq p_{g,2} \geq \dots \geq p_{g,n}$. Considering all possible sets $\mathcal{G}_g$ with fixed cardinality $n_g$, the optimal $\mathcal{G}^{*}_g$ with cardinality $n_g$ that maximizes the success probability is 
\begin{align}
    \mathcal{G}^{*}_g =\{1,2,\dots,n_g\}
\end{align}
which represents the set of $n_g$ workers having largest $p_{g,i}$ among all the workers. 
\end{lemma}
\begin{proof}
For a fixed integer $n_g$, we suppose $\mathcal{G}_1$ is the optimal set with cardinality $n_g$ where $i \notin \mathcal{G}_1$ and $1\leq i \leq n_g$. Thus, there exists a $j \in \mathcal{G}_1$ such that $j > n_g$. Then, we construct a set $\mathcal{G}_2 = (\mathcal{G}_1 \backslash \{j\}) \cup \{i\}$. The success probability of using the load allocation vector characterized by $\mathcal{G}_1$ can be written as 
\begin{align}
    &\mathbb{P}(T^{(\mathcal{G}_1)}(\vec{\mu}) \leq d) = \mathbb{P}( Q(\mathcal{G}_1) \geq a(\mathcal{G}_1)) \\
    = &p_{g,j}\mathbb{P}(Q(\mathcal{G}_1 \backslash \{j\}) \geq a(\mathcal{G}_1)-1)+(1-p_{g,j})\mathbb{P}(Q(\mathcal{G}_1 \backslash \{j\}) \geq a(\mathcal{G}_1)) \nonumber
\end{align}
where the first term is the success probability when worker $j$ is in the good state, and the second term is the success probability when worker $j$ is in bad state. 
Similarly, the success probability of using the load allocation vector characterized by $\mathcal{G}_2$ can be written as 
\begin{align}
    &\mathbb{P}(T^{(\mathcal{G}_2)}(\vec{\mu}) \leq d) = \mathbb{P}( Q(\mathcal{G}_2) \geq a(\mathcal{G}_2)) \\
    = &p_{g,i}\mathbb{P}(Q(\mathcal{G}_2 \backslash \{i\})\geq a(\mathcal{G}_2)-1)+(1-p_{g,i})\mathbb{P}(Q(\mathcal{G}_2 \backslash \{i\})\geq a(\mathcal{G}_2)) \nonumber,
\end{align}    
which can be further written as 
\begin{align}
    p_{g,i}\mathbb{P}(Q(\mathcal{G}_1 \backslash \{j\})\geq a(\mathcal{G}_1)-1)+(1-p_{g,i})\mathbb{P}(Q(\mathcal{G}_1 \backslash \{j\}) \geq a(\mathcal{G}_1)) \nonumber
\end{align}
since $\mathcal{G}_2 = (\mathcal{G}_1 \backslash \{j\}) \cup \{i\}$ and $a(\mathcal{G}_1) = a(\mathcal{G}_2)$.
Because $p_{g,i} \geq p_{g,j}$ and $\mathbb{P}(Q(\mathcal{G}_1 \backslash \{j\})\geq a(\mathcal{G}_1)-1) \geq \mathbb{P}(Q(\mathcal{G}_1 \backslash \{j\}) \geq a(\mathcal{G}_1))$, we have
\begin{align}
    &\mathbb{P}(T^{(\mathcal{G}_2)}(\vec{\mu}) \leq d)-\mathbb{P}(T^{(\mathcal{G}_1)}(\vec{\mu}) \leq d)\\
  =  & (p_{g,i}-p_{g,j})\{\mathbb{P}(Q(\mathcal{G}_1 \backslash \{j\})\geq a(\mathcal{G}_1)-1) - \mathbb{P}(Q(\mathcal{G}_1 \backslash \{j\}) \geq a(\mathcal{G}_1))\} \nonumber \\  \geq & 0 \nonumber
\end{align}
which is a contradiction. Thus, the optimal set $\mathcal{G}_g$ with fixed cardinality $n_g$ must include $i$ for all $1 \leq i \leq n_g$.
\end{proof}
By Lemma \ref{lemma:optimalset}, for a fixed cardinality $n_g$, the optimal $\mathcal{G}^*_g$ is the collection of $n_g$ workers having largest $p_{g,i}$ among all the workers. Therefore, to find the optimal load allocation vector, we can only focus on finding the optimal $n^{*}_g$. Since there are only $n$ choices for $n^{*}_g$ (i.e. $1,2,\dots,n$), the complexity of searching the optimal $n^{*}_g$ is linear in the number of workers $n$ which is much smaller than $2^n$ .

The following theorem shows that the computation strategy composed of the Lagrange coding scheme and the load allocation vector that is the solution of load allocation problem achieves the optimal timely computation throughput when the Markov model is known to the master.
\begin{theorem}\label{thm:opt_markov}
Assume the Markov model of the network is know to the master. Let the computation strategy $\eta^* =(\vec{g^*},\{\vec{\ell^*}_m\}^{\infty}_{m=1})$ be the computation strategy where $\vec{g^*}$ is the Lagrange coding scheme and $\{\vec{\ell^*}_m\}^{\infty}_{m=1}$ is given by solving load allocation problem. Then, $\eta^*$ achieves the optimal timely computation throughput.
\end{theorem}
\begin{proof}
We consider the computation of round $m$ and denote $N_m(d)$ as the indicator represents whether the computation is finished in time $d$ in round $m$ using an arbitrary computation strategy. Clearly, $N_m(d)$ is a Bernoulli random variable with parameter $\mathbb{P}(m)$ which denotes the success probability using this computation strategy in round $m$. Thus, $N_m(d)$ would contribute to the throughput with probability $\mathbb{P}(m)$. Since $\eta^*$ maximizes $\mathbb{P}(m)$ for all $m$, this strategy is optimal.
\end{proof}
Since the Markov model is unknown to the master in the original problem, the timely computation throughput achieved by $\eta^*$ gives us an upper bound. In the next section, we will show that this upper bound can be matched by using LEA.

\section{Optimality of LEA} 
Now, we show the optimality of LEA by the following theorem. 
\begin{theorem} \label{thm:opt_ea}
The proposed Lagrange Estimate-and-Allocate (LEA) strategy is optimal, i.e.,
\begin{align}
    R_{\text{LEA}}(d) =R^{*}(d) \ \text{almost surely},
\end{align}
where $R_{\text{LEA}}(d)$ denotes the timely computation throughput using the LEA strategy.
\end{theorem}
\begin{proof}
In order to prove Theorem \ref{thm:opt_ea}, we first state Lemma \ref{lemma:success} whose proof is moved to Appendix \ref{appendix:proof_success} for the purpose of readibility.
\begin{lemma} \label{lemma:success}
$\mathbb{P}_{\text{LEA}}(m)$ converges to $\mathbb{P}^{*}(m)$ as $m$ goes to infinity, where $\mathbb{P}^{*}(m)$ denotes the optimal success probability in round $m$ and $\mathbb{P}_{\text{LEA}}(m)$ denotes the success probability in round $m$ using the LEA strategy.
\end{lemma}
Before proving the optimality of LEA, we first define the following terms. We denote $N^{*}_m(d)$ as the indicator representing whether the computation is finished by time $d$ in round $m$ using the optimal computation strategy which maximizes the success probability in round $m$. Clearly, $N^{*}_m(d)$ is a Bernoulli random
variable with parameter $\mathbb{P}^{*}(m)$. Also, we denote $N_{\text{LEA},m}(d)$ as the indicator representing whether the computation is finished in time $d$ in round $m$ using LEA. Then, $N_{\text{LEA},m}(d)$ is a Bernoulli random variable with parameter $\mathbb{P}_{\text{LEA}}(m)$. We denote $R_{\text{LEA}}(d)$ as the timely computation throughput using LEA.

Now, we model the state of the whole system which includes all $n$ workers as a Markov chain. Since each worker has $2$ states (good or bad), there are a total of $2^n$ different states of the system. Without loss of generality, we index the states of the system as $\{1,2,\dots,2^n\}$. Clearly, the transition matrix of this Markov chain has all the entries larger than $0$. It implies that this Markov chain is irreducible. We denote $s(m)$ as the state of the system in round $m$. Also, $p^{*}_s$ is denoted as the success probability of state $s$ using the optimal computation strategy, i.e., $\mathbb{P}^{*}(m) = p^{*}_s$ if $s(m)=s$. By the Strong Law of Large Numbers and the Ergodic theorem, the optimal timely computation throughput $R^{*}(d)$ can be written as  
\begin{align}
    & R^{*}(d)  = \lim_{M \rightarrow \infty} \frac{\sum^M_{m=1}N^{*}_m(d)}{M}\\ = & \lim_{M \rightarrow \infty} \sum^{2^n}_{s=1} \frac{\sum_{m \geq 1:s(m)=s}N^{*}_m(d)}{V_i(M)}\frac{V_i(M)}{M}=  \sum^{2^n}_{s=1} p^{*}_s\frac{1}{\mathbb{E}_s[T_s]} \quad a.s.,  \nonumber
\end{align}
where the Ergodic theorem is formally stated as follows:
\begin{theorem*} [Ergodic Theorem]
If transition matrix $P$ of a Markov chain $(X_m)_{m \geq 0}$ is irreducible, then we have
\begin{align}
\lim_{m \rightarrow \infty} \frac{V_s(m)}{m} = \frac{1}{\mathbb{E}_s[T_s]}  \quad  a.s.
\end{align}  
where $V_s(m)$ is the number of visits to state $s$ up to round $m$ and $\mathbb{E}_s[T_s]$ is the expected return time to state $s$.
\end{theorem*}
By Lemma \ref{lemma:success}, for all $\epsilon > 0$, there exits $m(\epsilon)$ such that $\mathbb{P}_{\text{LEA}}(m) > \mathbb{P}^{*}(m)-\epsilon$ for all $m > m(\epsilon)$. Let $\tilde{N}_m(d)$ be the independent Bernoulli process with parameter $\mathbb{P}^{*}(m)-\epsilon$. We couple $N_{\text{LEA},m}(d)$ and $\tilde{N}_m(d)$ as follows. If $N_{\text{LEA},m}(d)=0$, then $\tilde{N}_m(d) = 0$. If $N_{\text{LEA},m}(d)=1$, then $\tilde{N}_m(d) = 1$ with probability $\frac{\mathbb{P}^{*}(m)-\epsilon}{\mathbb{P}_{\text{LEA}}(m)}$, and $\tilde{N}_m(d) = 0$ with probability $1-\frac{\mathbb{P}^{*}(m)-\epsilon}{\mathbb{P}_{\text{LEA}}(m)}$. Note that $\tilde{N}_m(d)$ is still marginally independent Bernoulli process of parameter $\mathbb{P}^{*}(m)-\epsilon$. Then, we have
\begin{align}
    & R_{\text{LEA}}(d) = \lim_{M \rightarrow \infty} \frac{\sum^M_{m=1}N_{\text{LEA},m}(d)}{M} \\
    & \geq \lim_{M \rightarrow \infty} \frac{\sum^M_{m=m(\epsilon)+1}N_{\text{LEA},m}(d)}{M}\\
    & \geq \lim_{M \rightarrow \infty}\frac{\sum^M_{m=m(\epsilon)+1}\tilde{N}_m(d)}{M}\\
    & =  \lim_{M \rightarrow \infty} \frac{1}{M}\sum^{2^n}_{s=1} \sum_{m\geq m(\epsilon)+1:s(m)=s}\tilde{N}_m(d)\\
     & =  \lim_{M \rightarrow \infty} \frac{M-m(\epsilon)}{M} \sum^{2^n}_{s=1} \frac{\sum_{m\geq m(\epsilon)+1:s(m)=s}\tilde{N}_m(d)}{V_s(M)-V_s(m(\epsilon))}\frac{V_s(M)-V_s(m(\epsilon))}{M-m(\epsilon)} \nonumber\\
     & = \sum^{2^n}_{s=1} (p^{*}_s-\epsilon) \frac{1}{\mathbb{E}_s[T_s]} = \sum^{2^n}_{s=1} p^{*}_s\frac{1}{\mathbb{E}_s[T_s]} - \sum^{2^n}_{s=1}\epsilon\frac{1}{\mathbb{E}_s[T_s]}\\ & = R^{*}(d) - \sum^{2^n}_{s=1}\epsilon\frac{1}{\mathbb{E}_s[T_s]} \quad a.s.
\end{align}
using the SLLN and the Ergodic theorem. Also, it is clear that $R_{\text{LEA}}(d) \leq R^{*}(d)$. Letting $\epsilon \rightarrow 0$, we have $R_{\text{EA}}(d) =R^{*}(d)$ which completes the proof. 
\end{proof}

\section{Experiments}
In this section, we present our results both from simulation studies as well as from experiments over Amazon EC2 cluster. 
\subsection{Numerical Analysis}
We now present numerical results evaluating the performance of the LEA strategy.

First, we call a computation strategy \emph{static} if this computation strategy assigns the loads to workers without considering their states in previous rounds. For comparison with LEA, we consider the following static computation strategy:

\textbf{Static Computation Strategy}: Prior to computation, Lagrange coding scheme is used for data encoding. In each round $m$, each worker $i$ is assigned a load $\ell_{m,i} \in \{\ell_g,\ell_b\}$ based on the stationary distributions of the underlying Markov model, in which we denote $(\pi_{g,i},\pi_{b,i})$ as stationary distribution of worker $i$. More specifically, for each worker $i$ in each round $m$, this strategy does assignment as follows:
\begin{align}
    \ell_{m,i} = 
    \begin{cases}
    \ell_g \ \text{with probability} \ \pi_{g,i}  \\ 
    \ell_b \ \text{with probability} \ \pi_{b,i}.
    \end{cases}
\end{align}
Note that whenever the total loads of the generated $\vec{\ell}_m$ is smaller than the minimum recovery threshold, then the strategy would do assignments again until the total loads of the generated $\vec{\ell}_m$ is greater than the minimum recovery threshold.

Since static computation strategies don't learn the dynamics of network, they can only do load assignments in a deterministic manner or randomly without using any history. Thus, the chosen static computation strategy which utilizes the stationary distributions of underlying Markov model is better than other static computation strategies in general.

Given deadline $d=1$ second in each round $m$, we consider a problem of evaluating a quadratic function $f_m(X_j)$ $=X_j^{\top}(X_j\vec{w}_m-\vec{y})$ over $n=15$ workers, where the dataset $X_1,X_2,\dots,X_{50} \in \mathbb{R}^{1000 \times 1000}$, $\vec{y} \in \mathbb{R}^{1000 \times 1}$ and $\vec{w}_m \in \mathbb{R}^{1000 \times 1}$ which is the input vector in round $m$. Each worker stores $r=10$ encoded data chunks using Lagrange coding scheme. In such setting, we have the optimal recovery threshold $K^* = 99$ for both LEA and the static computation strategy. 

For simulations, we let $p_{g \rightarrow g,i} = p_{g \rightarrow g}, p_{b \rightarrow b,i} = p_{b \rightarrow b}$ for all $i$, and consider the following four scenarios:\\
\textbf{Scenario 1:} $(\mu_g,\mu_b)=(10,3)$, $(p_{g \rightarrow g},p_{b \rightarrow b}) = (0.8,0.8)$ and the corresponding stationary probabilities $(p_{g},p_{b}) = (0.5,0.5)$.\\
\textbf{Scenario 2:} $(\mu_g,\mu_b)=(10,3)$, $(p_{g \rightarrow g},p_{b \rightarrow b}) = (0.8,0.7)$ and the corresponding stationary probabilities $(p_{g},p_{b}) = (0.6,0.4)$.\\
\textbf{Scenario 3:} $(\mu_g,\mu_b)=(10,3)$, $(p_{g \rightarrow g},p_{b \rightarrow b}) = (0.8,0.533)$ and the corresponding stationary probabilities $(p_{g},p_{b}) = (0.7,0.3)$.\\
\textbf{Scenario 4:} $(\mu_g,\mu_b)=(10,3)$, $(p_{g \rightarrow g},p_{b \rightarrow b}) = (0.9,0.6)$ and the corresponding stationary probabilities $(p_{g},p_{b}) = (0.8,0.2)$.

Fig. \ref{fig:numerical} illustrate the performance comparison for LEA and the static computation strategy. We make the following conclusions from the results:\\
$\bullet$ LEA increases substantial improvement in terms of the timely computation throughput. Over the four scenarios, LEA improves the static computation strategy by $1.38 \times \sim 17.5 \times$.\\
$\bullet$ The timely computation throughput improvements over the static computation strategy become more significant as the stationary probability $p_g$ decreases. When $p_g$ is small, the workers would be in the bad state more probably in the long run. In this sense, the static computation strategy assigns loads to the workers in a more pessimistic way. However, there is temporal correlation of computation speeds which the static computation strategy doesn't take into account. Thus, although $p_g$ is small, LEA can achieve much higher timely computation throughput which demonstrates that LEA can adapt to the dynamics of network well.
\subsection{Experiments using Amazon EC2 machines}
\begin{figure}[t]
  \centering
    \includegraphics[width = \columnwidth]{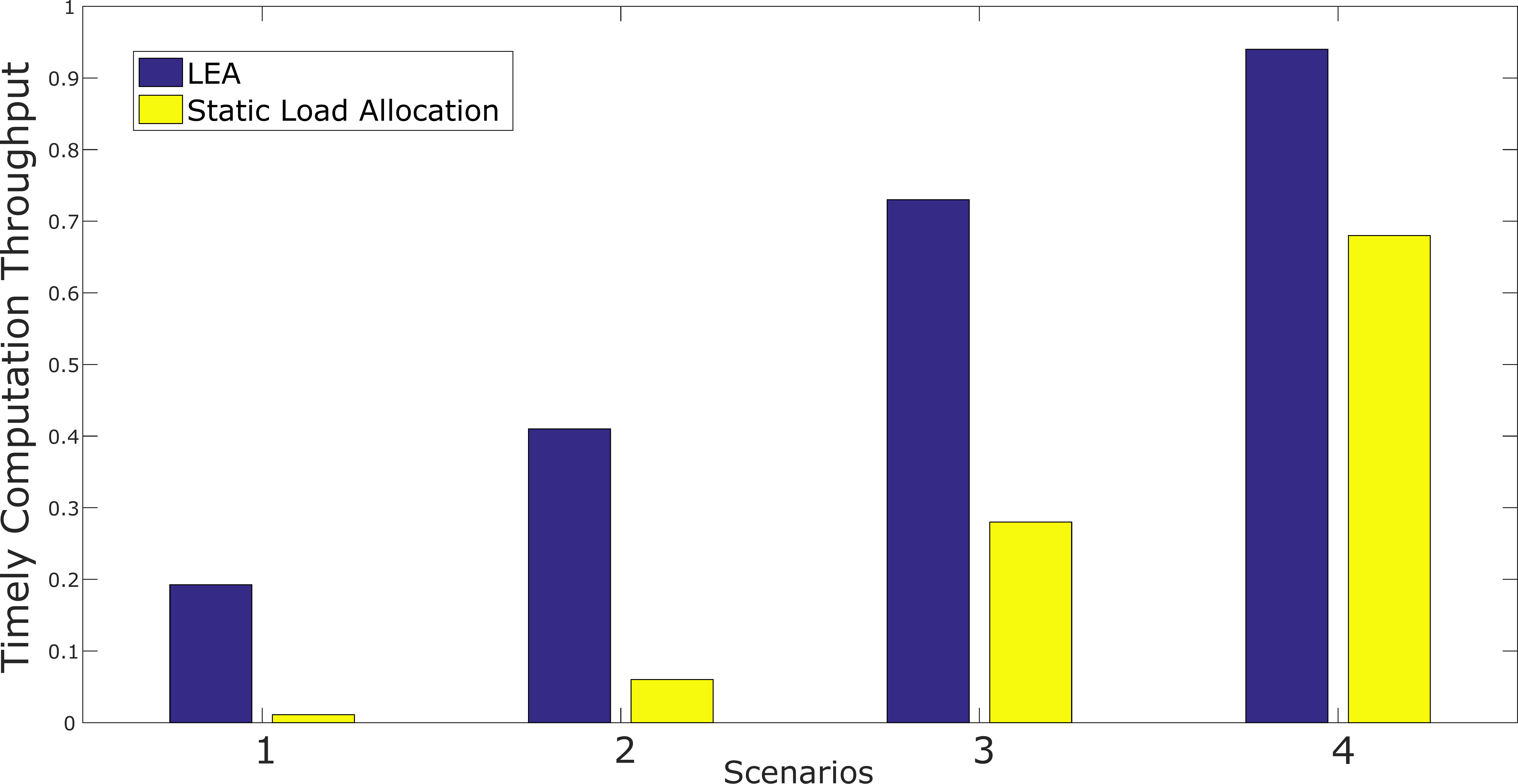}
\caption{Numerical Results}
\label{fig:numerical}
\end{figure}
Before showing the experimental results, we first introduce CPU credits \cite{amazonEC2} which can boost T2 and T3 instances above baseline performance. For a \texttt{t2.micro} instance, as shown in Fig. \ref{fig:credit_speed}, there is a $10$ times difference between baseline performance and burstable performance, i.e., a burst \texttt{t2.micro} instance has computing speed $10$ times faster. The baseline performance and ability to burst are governed by CPU credits. More details of CPU credits and burstable performance can be found in \cite{amazonEC2}.

We ran the master node over \texttt{m4.xlarge} instance and all workers over \texttt{t2.micro} instances. We implemented two computation strategies in python, and used MPI4py \cite{dalcin2011parallel} for message passing between instances. Before starting computations, each worker stores a certain amount of data in its local memory. In round $m$, having received function $f_m$ from the master, each worker computes the assigned computation using the stored data, and sends it back to the master asynchronously using \texttt{Isend()}. As soon as the master gathers enough results from the workers, it computes the evaluations for the function $f_m$.

Given deadline $d$ seconds in each round $m$, we consider a problem of evaluating a linear function $f_m(X_j)$ $=X_j^{\top}B_m$ over $n=15$ workers, where the datasets $\{X_j\}^k_{j=1}$'s are real matrices with certain dimensions, and $B_m \in \mathbb{R}^{3000 \times 3000}$ is the input matrix. Each worker stores $r=10$ encoded data chunks using Lagrange coding scheme. In particular, in each round, the computation request's arrival time is shift-exponential random variable which is the sum of a constant $T_c=30$ and an exponential random variable with mean $\lambda$. In this setting, we have the optimal recovery threshold $K^* = 50$ for both LEA and the static computation strategy. Since the Markov model is unknown (and indeed even the type of the underlying stochastic process determining the states of the workers in the cloud is not known), to compare with the LEA strategy, we consider a static computation strategy that each worker is assigned to $\ell_g$ or $\ell_b$ number of evaluations with equal probability in each round. For experiments, we consider the following six scenarios:\\
\textbf{Scenario 1:} Size of $X_j = 25 \times 3000$, $k=120$, $\lambda =10$ and $d =2.5$. \\
\textbf{Scenario 2:} Size of $X_j = 25 \times 3000$, $k=120$, $\lambda =30$ and $d =2.5$. \\
\textbf{Scenario 3:} Size of $X_j = 30 \times 3000$, $k=100$, $\lambda =10$ and $d =3$. \\
\textbf{Scenario 4:} Size of $X_j = 30 \times 3000$, $k=100$, $\lambda =30$ and $d =3$. \\
\textbf{Scenario 5:} Size of $X_j = 60 \times 3000$, $k=50$, $\lambda =10$ and $d =6$. \\
\textbf{Scenario 6:} Size of $X_j = 60 \times 3000$, $k=50$, $\lambda =30$ and $d =6$. 
\begin{figure}[t]
  \centering
    \includegraphics[width = \columnwidth]{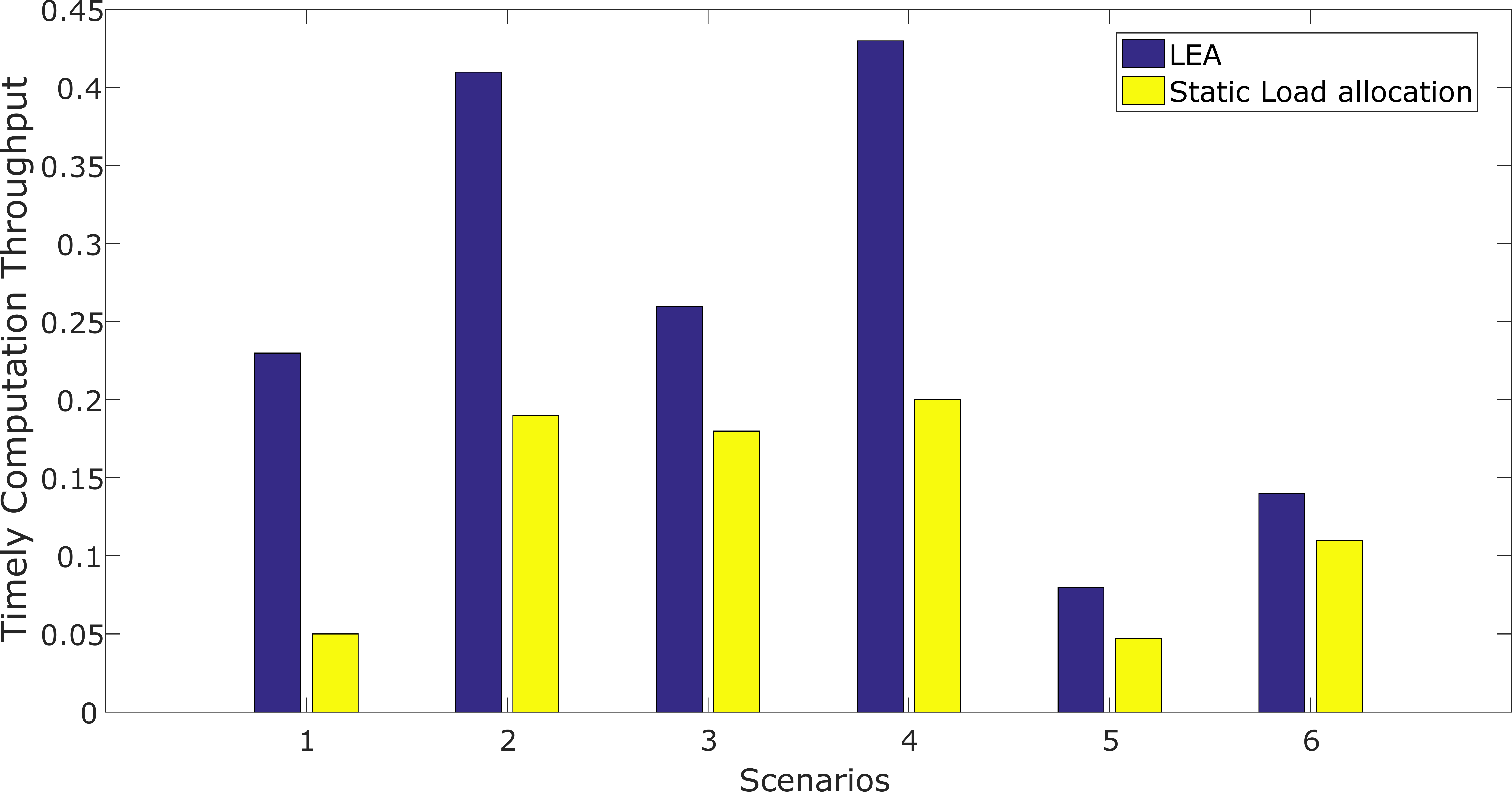}
\caption{Experimental evaluations over $15$ \texttt{t2.micro} instances in Amazon EC2. Compared with the static load allocation strategy, LEA improves the timely computation throughput by $1.27 \times \sim 6.5 \times$.}
\label{fig:experiments}
\end{figure}

Fig. \ref{fig:experiments} provides a performance comparison of LEA with the static load allocation strategy for the six scenarios. From the results, we found that LEA provides substantial improvement in terms of the timely computation throughput. Over the six scenarios, LEA increases the static computation strategy by $1.27 \times \sim 6.5 \times$.

\section{Conclusion}
 Motivated by high variability of computing resources in modern distributed computing systems and increasing demand for timely event-driven services with deadline constraints, we consider the problem of dynamic computation load allocation over a coded computing framework. We propose an optimal dynamic computation strategy Lagrange Estimate and Allocate, LEA, which is composed of utilizing the Lagrange coding scheme for data encoding and assigning computation loads based on the estimated state of the network, which is done by estimating the transition probabilities of an underlying Markov model for the system's state from observing the past events at each time step.
In the end, we show that compared to the static computation strategy, LEA increases the timely computation throughput by $1.38 \times \sim 17.5 \times$ in simulations and by $1.27 \times \sim 6.5 \times$ in Amazon EC2 clusters.

At a conceptual level, this paper has some interesting comparisons/connections with \cite{krishnasamy2018augmenting}. Under wireless networks, \cite{krishnasamy2018augmenting} investigates how to turn base stations on or off, in order to adapt to the unknown load arrival and channel statistics. Under cloud computing networks, our paper focuses on how to do the computation load assignment in order to adapt to unknown computing networks. So, at a high-level, the corresponding scheduling problems can be seen as dual of each other: \cite{krishnasamy2018augmenting} assigns base stations to good (on) or bad (off) states in order to meet the demands, while our goal is to assign the computation loads in order to optimally exploit the (unknown) state of the workers. However, we also point out that the setting and objective of the two papers are quite different. We consider cloud computing platforms and focus on the timely computation throughput, which is very different from \cite{krishnasamy2018augmenting}. 
 Another difference is in the proof techniques to show the optimality of the proposed algorithms. The Lyapunov arguments for the adaptive scheme used in \cite{krishnasamy2018augmenting} is quite different from our approach.
 %

\section{Acknowledgment}
This material is based upon work supported by Defense Advanced Research Projects Agency (DARPA) under Contract No. HR001117C0053, ARO award W911NF1810400, NSF grants CCF-1703575, ONR Award No. N00014-16-1-2189, and CCF-1763673. The views, opinions, and/or findings expressed are those of the author(s) and should not be interpreted as representing the official views or policies of the Department of Defense or the U.S. Government. This work is also in part supported by ONR award
N000141612189 and NSF Grants CCF-1703575 and NeTS-1419632
and the UC Office of President under grant No. LFR-18-548175.

\bibliographystyle{ACM-Reference-Format}
\bibliography{references.bib} 


\begin{thebibliography}{00}


\ifx \showCODEN    \undefined \def \showCODEN     #1{\unskip}     \fi
\ifx \showDOI      \undefined \def \showDOI       #1{#1}\fi
\ifx \showISBNx    \undefined \def \showISBNx     #1{\unskip}     \fi
\ifx \showISBNxiii \undefined \def \showISBNxiii  #1{\unskip}     \fi
\ifx \showISSN     \undefined \def \showISSN      #1{\unskip}     \fi
\ifx \showLCCN     \undefined \def \showLCCN      #1{\unskip}     \fi
\ifx \shownote     \undefined \def \shownote      #1{#1}          \fi
\ifx \showarticletitle \undefined \def \showarticletitle #1{#1}   \fi
\ifx \showURL      \undefined \def \showURL       {\relax}        \fi
\providecommand\bibfield[2]{#2}
\providecommand\bibinfo[2]{#2}
\providecommand\natexlab[1]{#1}
\providecommand\showeprint[2][]{arXiv:#2}

\bibitem[\protect\citeauthoryear{Ananthanarayanan, Ghodsi, Shenker, and
  Stoica}{Ananthanarayanan et~al\mbox{.}}{2013}]%
        {ananthanarayanan2013effective}
\bibfield{author}{\bibinfo{person}{Ganesh Ananthanarayanan},
  \bibinfo{person}{Ali Ghodsi}, \bibinfo{person}{Scott Shenker}, {and}
  \bibinfo{person}{Ion Stoica}.} \bibinfo{year}{2013}\natexlab{}.
\newblock \showarticletitle{Effective Straggler Mitigation: Attack of the
  Clones.}. In \bibinfo{booktitle}{{\em NSDI}}, Vol.~\bibinfo{volume}{13}.
  \bibinfo{pages}{185--198}.
\newblock


\bibitem[\protect\citeauthoryear{Arabnejad, Bubendorfer, and Ng}{Arabnejad
  et~al\mbox{.}}{2017}]%
        {arabnejad2017scheduling}
\bibfield{author}{\bibinfo{person}{Vahid Arabnejad}, \bibinfo{person}{Kris
  Bubendorfer}, {and} \bibinfo{person}{Bryan Ng}.}
  \bibinfo{year}{2017}\natexlab{}.
\newblock \showarticletitle{Scheduling deadline constrained scientific
  workflows on dynamically provisioned cloud resources}.
\newblock \bibinfo{journal}{{\em Future Generation Computer Systems\/}}
  \bibinfo{volume}{75} (\bibinfo{year}{2017}), \bibinfo{pages}{348--364}.
\newblock


\bibitem[\protect\citeauthoryear{Baccelli, Massey, and Towsley}{Baccelli
  et~al\mbox{.}}{1989}]%
        {baccelli1989acyclic}
\bibfield{author}{\bibinfo{person}{Fran{\c{c}}ois Baccelli},
  \bibinfo{person}{William~A Massey}, {and} \bibinfo{person}{Don Towsley}.}
  \bibinfo{year}{1989}\natexlab{}.
\newblock \showarticletitle{Acyclic fork-join queuing networks}.
\newblock \bibinfo{journal}{{\em Journal of the ACM (JACM)\/}}
  \bibinfo{volume}{36}, \bibinfo{number}{3} (\bibinfo{year}{1989}).
\newblock


\bibitem[\protect\citeauthoryear{Bitar, Parag, and El~Rouayheb}{Bitar
  et~al\mbox{.}}{2017}]%
        {bitar2017minimizing}
\bibfield{author}{\bibinfo{person}{Rawad Bitar}, \bibinfo{person}{Parimal
  Parag}, {and} \bibinfo{person}{Salim El~Rouayheb}.}
  \bibinfo{year}{2017}\natexlab{}.
\newblock \showarticletitle{Minimizing latency for secure distributed
  computing}. In \bibinfo{booktitle}{{\em Information Theory (ISIT), 2017 IEEE
  International Symposium on}}. IEEE, \bibinfo{pages}{2900--2904}.
\newblock


\bibitem[\protect\citeauthoryear{Chen, Charles, Papailiopoulos,
  et~al\mbox{.}}{Chen et~al\mbox{.}}{2018}]%
        {chen2018draco}
\bibfield{author}{\bibinfo{person}{Lingjiao Chen}, \bibinfo{person}{Zachary
  Charles}, \bibinfo{person}{Dimitris Papailiopoulos}, {et~al\mbox{.}}}
  \bibinfo{year}{2018}\natexlab{}.
\newblock \showarticletitle{DRACO: Robust Distributed Training via Redundant
  Gradients}.
\newblock \bibinfo{journal}{{\em arXiv preprint arXiv:1803.09877\/}}
  (\bibinfo{year}{2018}).
\newblock


\bibitem[\protect\citeauthoryear{Dai and Lin}{Dai and Lin}{2005}]%
        {dai2005maximum}
\bibfield{author}{\bibinfo{person}{Jim~G Dai} {and} \bibinfo{person}{Wuqin
  Lin}.} \bibinfo{year}{2005}\natexlab{}.
\newblock \showarticletitle{Maximum pressure policies in stochastic processing
  networks}.
\newblock \bibinfo{journal}{{\em Operations Research\/}} \bibinfo{volume}{53},
  \bibinfo{number}{2} (\bibinfo{year}{2005}).
\newblock


\bibitem[\protect\citeauthoryear{Dalcin, Paz, Kler, and Cosimo}{Dalcin
  et~al\mbox{.}}{2011}]%
        {dalcin2011parallel}
\bibfield{author}{\bibinfo{person}{Lisandro~D Dalcin},
  \bibinfo{person}{Rodrigo~R Paz}, \bibinfo{person}{Pablo~A Kler}, {and}
  \bibinfo{person}{Alejandro Cosimo}.} \bibinfo{year}{2011}\natexlab{}.
\newblock \showarticletitle{Parallel distributed computing using Python}.
\newblock \bibinfo{journal}{{\em Advances in Water Resources\/}}
  \bibinfo{volume}{34}, \bibinfo{number}{9} (\bibinfo{year}{2011}),
  \bibinfo{pages}{1124--1139}.
\newblock


\bibitem[\protect\citeauthoryear{Dutta, Cadambe, and Grover}{Dutta
  et~al\mbox{.}}{2016}]%
        {dutta2016short}
\bibfield{author}{\bibinfo{person}{Sanghamitra Dutta}, \bibinfo{person}{Viveck
  Cadambe}, {and} \bibinfo{person}{Pulkit Grover}.}
  \bibinfo{year}{2016}\natexlab{}.
\newblock \showarticletitle{Short-dot: Computing large linear transforms
  distributedly using coded short dot products}. In \bibinfo{booktitle}{{\em
  Advances In Neural Information Processing Systems}}.
  \bibinfo{pages}{2100--2108}.
\newblock


\bibitem[\protect\citeauthoryear{EC2:}{EC2:}{[n. d.]}]%
        {amazonEC2}
\bibfield{author}{\bibinfo{person}{Amazon EC2:}.} \bibinfo{year}{[n.
  d.]}\natexlab{}.
\newblock \bibinfo{booktitle}{{\em https://docs.aws.amazon.com/ec2/}}.
\newblock


\bibitem[\protect\citeauthoryear{Eryilmaz and Srikant}{Eryilmaz and
  Srikant}{2007}]%
        {eryilmaz2007fair}
\bibfield{author}{\bibinfo{person}{Atilla Eryilmaz} {and} \bibinfo{person}{R
  Srikant}.} \bibinfo{year}{2007}\natexlab{}.
\newblock \showarticletitle{Fair resource allocation in wireless networks using
  queue-length-based scheduling and congestion control}.
\newblock \bibinfo{journal}{{\em IEEE/ACM Transactions on Networking (TON)\/}}
  \bibinfo{volume}{15}, \bibinfo{number}{6} (\bibinfo{year}{2007}),
  \bibinfo{pages}{1333--1344}.
\newblock


\bibitem[\protect\citeauthoryear{Eryilmaz, Srikant, and Perkins}{Eryilmaz
  et~al\mbox{.}}{2005}]%
        {eryilmaz2005stable}
\bibfield{author}{\bibinfo{person}{Atilla Eryilmaz},
  \bibinfo{person}{Rayadurgam Srikant}, {and} \bibinfo{person}{James~R
  Perkins}.} \bibinfo{year}{2005}\natexlab{}.
\newblock \showarticletitle{Stable scheduling policies for fading wireless
  channels}.
\newblock \bibinfo{journal}{{\em IEEE/ACM Transactions on Networking\/}}
  \bibinfo{volume}{13}, \bibinfo{number}{2} (\bibinfo{year}{2005}),
  \bibinfo{pages}{411--424}.
\newblock


\bibitem[\protect\citeauthoryear{Hoseinnejhad and Navimipour}{Hoseinnejhad and
  Navimipour}{2017}]%
        {hoseinnejhad2017deadline}
\bibfield{author}{\bibinfo{person}{Mina Hoseinnejhad} {and}
  \bibinfo{person}{Nima~Jafari Navimipour}.} \bibinfo{year}{2017}\natexlab{}.
\newblock \showarticletitle{Deadline constrained task scheduling in the cloud
  computing using a discrete firefly algorithm}.
\newblock \bibinfo{journal}{{\em INTERNATIONAL JOURNAL OF NEXT-GENERATION
  COMPUTING\/}} \bibinfo{volume}{8}, \bibinfo{number}{3}
  (\bibinfo{year}{2017}).
\newblock


\bibitem[\protect\citeauthoryear{Hou, Borkar, and Kumar}{Hou
  et~al\mbox{.}}{2009}]%
        {timelyThroughput}
\bibfield{author}{\bibinfo{person}{I.~. Hou}, \bibinfo{person}{V. Borkar},
  {and} \bibinfo{person}{P.~R. Kumar}.} \bibinfo{year}{2009}\natexlab{}.
\newblock \showarticletitle{A Theory of QoS for Wireless}. In
  \bibinfo{booktitle}{{\em IEEE INFOCOM 2009}}. \bibinfo{pages}{486--494}.
\newblock
\showISSN{0743-166X}
\showDOI{%
\url{https://doi.org/10.1109/INFCOM.2009.5061954}}


\bibitem[\protect\citeauthoryear{Krishnasamy, Akhil, Arapostathis, Sundaresan,
  and Shakkottai}{Krishnasamy et~al\mbox{.}}{2018}]%
        {krishnasamy2018augmenting}
\bibfield{author}{\bibinfo{person}{Subhashini Krishnasamy}, \bibinfo{person}{PT
  Akhil}, \bibinfo{person}{Ari Arapostathis}, \bibinfo{person}{Rajesh
  Sundaresan}, {and} \bibinfo{person}{Sanjay Shakkottai}.}
  \bibinfo{year}{2018}\natexlab{}.
\newblock \showarticletitle{Augmenting max-weight with explicit learning for
  wireless scheduling with switching costs}.
\newblock \bibinfo{journal}{{\em IEEE/ACM Transactions on Networking\/}}
  \bibinfo{volume}{26}, \bibinfo{number}{6} (\bibinfo{year}{2018}),
  \bibinfo{pages}{2501--2514}.
\newblock


\bibitem[\protect\citeauthoryear{Kwok and Ahmad}{Kwok and Ahmad}{1999}]%
        {kwok1999static}
\bibfield{author}{\bibinfo{person}{Yu-Kwong Kwok} {and} \bibinfo{person}{Ishfaq
  Ahmad}.} \bibinfo{year}{1999}\natexlab{}.
\newblock \showarticletitle{Static scheduling algorithms for allocating
  directed task graphs to multiprocessors}.
\newblock \bibinfo{journal}{{\em ACM Computing Surveys (CSUR)\/}}
  \bibinfo{volume}{31}, \bibinfo{number}{4} (\bibinfo{year}{1999}),
  \bibinfo{pages}{406--471}.
\newblock


\bibitem[\protect\citeauthoryear{Lashgari and Avestimehr}{Lashgari and
  Avestimehr}{2013}]%
        {lashgari2013timely}
\bibfield{author}{\bibinfo{person}{Sina Lashgari} {and}
  \bibinfo{person}{A~Salman Avestimehr}.} \bibinfo{year}{2013}\natexlab{}.
\newblock \showarticletitle{Timely throughput of heterogeneous wireless
  networks: Fundamental limits and algorithms}.
\newblock \bibinfo{journal}{{\em IEEE Transactions on Information Theory\/}}
  \bibinfo{volume}{59}, \bibinfo{number}{12} (\bibinfo{year}{2013}),
  \bibinfo{pages}{8414--8433}.
\newblock


\bibitem[\protect\citeauthoryear{Lee, Lam, Pedarsani, Papailiopoulos, and
  Ramchandran}{Lee et~al\mbox{.}}{2018}]%
        {lee2018speeding}
\bibfield{author}{\bibinfo{person}{Kangwook Lee}, \bibinfo{person}{Maximilian
  Lam}, \bibinfo{person}{Ramtin Pedarsani}, \bibinfo{person}{Dimitris
  Papailiopoulos}, {and} \bibinfo{person}{Kannan Ramchandran}.}
  \bibinfo{year}{2018}\natexlab{}.
\newblock \showarticletitle{Speeding up distributed machine learning using
  codes}.
\newblock \bibinfo{journal}{{\em IEEE Transactions on Information Theory\/}}
  \bibinfo{volume}{64}, \bibinfo{number}{3} (\bibinfo{year}{2018}),
  \bibinfo{pages}{1514--1529}.
\newblock


\bibitem[\protect\citeauthoryear{Li, Maddah-Ali, and Avestimehr}{Li
  et~al\mbox{.}}{2017}]%
        {li2017coding}
\bibfield{author}{\bibinfo{person}{Songze Li}, \bibinfo{person}{Mohammad~Ali
  Maddah-Ali}, {and} \bibinfo{person}{A~Salman Avestimehr}.}
  \bibinfo{year}{2017}\natexlab{}.
\newblock \showarticletitle{Coding for distributed fog computing}.
\newblock \bibinfo{journal}{{\em IEEE Communications Magazine\/}}
  \bibinfo{volume}{55}, \bibinfo{number}{4} (\bibinfo{year}{2017}),
  \bibinfo{pages}{34--40}.
\newblock


\bibitem[\protect\citeauthoryear{Li, Maddah-Ali, Yu, and Avestimehr}{Li
  et~al\mbox{.}}{2018}]%
        {li2018fundamental}
\bibfield{author}{\bibinfo{person}{Songze Li}, \bibinfo{person}{Mohammad~Ali
  Maddah-Ali}, \bibinfo{person}{Qian Yu}, {and} \bibinfo{person}{A~Salman
  Avestimehr}.} \bibinfo{year}{2018}\natexlab{}.
\newblock \showarticletitle{A fundamental tradeoff between computation and
  communication in distributed computing}.
\newblock \bibinfo{journal}{{\em IEEE Transactions on Information Theory\/}}
  \bibinfo{volume}{64}, \bibinfo{number}{1} (\bibinfo{year}{2018}),
  \bibinfo{pages}{109--128}.
\newblock


\bibitem[\protect\citeauthoryear{Maguluri, Srikant, and Ying}{Maguluri
  et~al\mbox{.}}{2012}]%
        {maguluri2012stochastic}
\bibfield{author}{\bibinfo{person}{Siva~Theja Maguluri}, \bibinfo{person}{R
  Srikant}, {and} \bibinfo{person}{Lei Ying}.} \bibinfo{year}{2012}\natexlab{}.
\newblock \showarticletitle{Stochastic models of load balancing and scheduling
  in cloud computing clusters}. In \bibinfo{booktitle}{{\em INFOCOM, 2012
  Proceedings IEEE}}. IEEE, \bibinfo{pages}{702--710}.
\newblock


\bibitem[\protect\citeauthoryear{Neely, Modiano, and Rohrs}{Neely
  et~al\mbox{.}}{2005}]%
        {neely2005dynamic}
\bibfield{author}{\bibinfo{person}{Michael~J Neely}, \bibinfo{person}{Eytan
  Modiano}, {and} \bibinfo{person}{Charles~E Rohrs}.}
  \bibinfo{year}{2005}\natexlab{}.
\newblock \showarticletitle{Dynamic power allocation and routing for
  time-varying wireless networks}.
\newblock \bibinfo{journal}{{\em IEEE Journal on Selected Areas in
  Communications\/}} \bibinfo{volume}{23}, \bibinfo{number}{1}
  (\bibinfo{year}{2005}), \bibinfo{pages}{89--103}.
\newblock


\bibitem[\protect\citeauthoryear{Pedarsani, Walrand, and Zhong}{Pedarsani
  et~al\mbox{.}}{2017}]%
        {pedarsani2017robust}
\bibfield{author}{\bibinfo{person}{Ramtin Pedarsani}, \bibinfo{person}{Jean
  Walrand}, {and} \bibinfo{person}{Yuan Zhong}.}
  \bibinfo{year}{2017}\natexlab{}.
\newblock \showarticletitle{Robust scheduling for flexible processing
  networks}.
\newblock \bibinfo{journal}{{\em Advances in Applied Probability\/}}
  \bibinfo{volume}{49} (\bibinfo{year}{2017}).
\newblock


\bibitem[\protect\citeauthoryear{Prakash, Reisizadeh, Pedarsani, and
  Avestimehr}{Prakash et~al\mbox{.}}{2018}]%
        {prakash2018coded}
\bibfield{author}{\bibinfo{person}{Saurav Prakash},
  \bibinfo{person}{Amirhossein Reisizadeh}, \bibinfo{person}{Ramtin Pedarsani},
  {and} \bibinfo{person}{Salman Avestimehr}.} \bibinfo{year}{2018}\natexlab{}.
\newblock \showarticletitle{Coded computing for distributed graph analytics}.
  In \bibinfo{booktitle}{{\em 2018 IEEE International Symposium on Information
  Theory (ISIT)}}. IEEE, \bibinfo{pages}{1221--1225}.
\newblock


\bibitem[\protect\citeauthoryear{Reisizadeh, Prakash, Pedarsani, and
  Avestimehr}{Reisizadeh et~al\mbox{.}}{2017}]%
        {reisizadeh2017coded}
\bibfield{author}{\bibinfo{person}{Amirhossein Reisizadeh},
  \bibinfo{person}{Saurav Prakash}, \bibinfo{person}{Ramtin Pedarsani}, {and}
  \bibinfo{person}{Salman Avestimehr}.} \bibinfo{year}{2017}\natexlab{}.
\newblock \showarticletitle{Coded computation over heterogeneous clusters}. In
  \bibinfo{booktitle}{{\em Information Theory (ISIT), 2017 IEEE International
  Symposium on}}. IEEE, \bibinfo{pages}{2408--2412}.
\newblock


\bibitem[\protect\citeauthoryear{Tandon, Lei, Dimakis, and
  Karampatziakis}{Tandon et~al\mbox{.}}{2017}]%
        {tandon2017gradient}
\bibfield{author}{\bibinfo{person}{Rashish Tandon}, \bibinfo{person}{Qi Lei},
  \bibinfo{person}{Alexandros~G Dimakis}, {and} \bibinfo{person}{Nikos
  Karampatziakis}.} \bibinfo{year}{2017}\natexlab{}.
\newblock \showarticletitle{Gradient coding: Avoiding stragglers in distributed
  learning}. In \bibinfo{booktitle}{{\em International Conference on Machine
  Learning}}. \bibinfo{pages}{3368--3376}.
\newblock


\bibitem[\protect\citeauthoryear{Tassiulas and Ephremides}{Tassiulas and
  Ephremides}{1992}]%
        {tassiulas1992stability}
\bibfield{author}{\bibinfo{person}{Leandros Tassiulas} {and}
  \bibinfo{person}{Anthony Ephremides}.} \bibinfo{year}{1992}\natexlab{}.
\newblock \showarticletitle{Stability properties of constrained queueing
  systems and scheduling policies for maximum throughput in multihop radio
  networks}.
\newblock \bibinfo{journal}{{\em IEEE transactions on automatic control\/}}
  \bibinfo{volume}{37}, \bibinfo{number}{12} (\bibinfo{year}{1992}),
  \bibinfo{pages}{1936--1948}.
\newblock


\bibitem[\protect\citeauthoryear{Topcuoglu, Hariri, and Wu}{Topcuoglu
  et~al\mbox{.}}{2002}]%
        {topcuoglu2002performance}
\bibfield{author}{\bibinfo{person}{Haluk Topcuoglu}, \bibinfo{person}{Salim
  Hariri}, {and} \bibinfo{person}{Min-you Wu}.}
  \bibinfo{year}{2002}\natexlab{}.
\newblock \showarticletitle{Performance-effective and low-complexity task
  scheduling for heterogeneous computing}.
\newblock \bibinfo{journal}{{\em IEEE transactions on parallel and distributed
  systems\/}} \bibinfo{volume}{13}, \bibinfo{number}{3} (\bibinfo{year}{2002}),
  \bibinfo{pages}{260--274}.
\newblock


\bibitem[\protect\citeauthoryear{Yang, Pedarsani, and Avestimehr}{Yang
  et~al\mbox{.}}{2018}]%
        {Yang1806:Communication}
\bibfield{author}{\bibinfo{person}{Chien-Sheng Yang}, \bibinfo{person}{Ramtin
  Pedarsani}, {and} \bibinfo{person}{Salman Avestimehr}.}
  \bibinfo{year}{2018}\natexlab{}.
\newblock \showarticletitle{{Communication-Aware} Scheduling of Serial Tasks
  for Dispersed Computing}. In \bibinfo{booktitle}{{\em 2018 IEEE International
  Symposium on Information Theory (ISIT) (ISIT'2018)}}. \bibinfo{address}{Vail,
  USA}.
\newblock


\bibitem[\protect\citeauthoryear{Yu, Li, Raviv, Mousavi, Soltanolkotabi, and
  Avestimehr}{Yu et~al\mbox{.}}{2019}]%
        {yu2019lagrange}
\bibfield{author}{\bibinfo{person}{Qian Yu}, \bibinfo{person}{Songze Li},
  \bibinfo{person}{Netanel Raviv}, \bibinfo{person}{Seyed~Mohammadreza
  Mousavi}, \bibinfo{person}{Mahdi Soltanolkotabi}, {and}
  \bibinfo{person}{A~Salman Avestimehr}.} \bibinfo{year}{2019}\natexlab{}.
\newblock \showarticletitle{Lagrange Coded Computing: Optimal Design for
  Resiliency, Security and Privacy}. In \bibinfo{booktitle}{{\em Artificial
  Intelligence and Statistics}}.
\newblock


\bibitem[\protect\citeauthoryear{Yu, Maddah-Ali, and Avestimehr}{Yu
  et~al\mbox{.}}{2017}]%
        {yu2017polynomial}
\bibfield{author}{\bibinfo{person}{Qian Yu}, \bibinfo{person}{Mohammad
  Maddah-Ali}, {and} \bibinfo{person}{Salman Avestimehr}.}
  \bibinfo{year}{2017}\natexlab{}.
\newblock \showarticletitle{Polynomial codes: an optimal design for
  high-dimensional coded matrix multiplication}. In \bibinfo{booktitle}{{\em
  Advances in Neural Information Processing Systems}}.
  \bibinfo{pages}{4403--4413}.
\newblock


\bibitem[\protect\citeauthoryear{Zaharia, Konwinski, Joseph, Katz, and
  Stoica}{Zaharia et~al\mbox{.}}{2008}]%
        {zaharia2008improving}
\bibfield{author}{\bibinfo{person}{Matei Zaharia}, \bibinfo{person}{Andy
  Konwinski}, \bibinfo{person}{Anthony~D Joseph}, \bibinfo{person}{Randy~H
  Katz}, {and} \bibinfo{person}{Ion Stoica}.} \bibinfo{year}{2008}\natexlab{}.
\newblock \showarticletitle{Improving MapReduce performance in heterogeneous
  environments.}. In \bibinfo{booktitle}{{\em Osdi}}, Vol.~\bibinfo{volume}{8}.
  \bibinfo{pages}{7}.
\newblock


\bibitem[\protect\citeauthoryear{Zheng and Sakellariou}{Zheng and
  Sakellariou}{2013}]%
        {zheng2013stochastic}
\bibfield{author}{\bibinfo{person}{Wei Zheng} {and} \bibinfo{person}{Rizos
  Sakellariou}.} \bibinfo{year}{2013}\natexlab{}.
\newblock \showarticletitle{Stochastic DAG scheduling using a Monte Carlo
  approach}.
\newblock \bibinfo{journal}{{\it J. Parallel and Distrib. Comput.}}
  \bibinfo{volume}{73}, \bibinfo{number}{12} (\bibinfo{year}{2013}),
  \bibinfo{pages}{1673--1689}.
\newblock


\end{thebibliography}
\section*{appendix}
\appendix
\section{Proof of Lemma \ref{lemma:mon}}\label{appendix:proof_mon}
Given an outcome of $\vec{\mu}$, we denote $Y(d,\vec{\mu},\vec{\ell})$ as the total number of results sent back to the master in time $d$ using the load allocation vector $\vec{\ell}$. We define two events $A \triangleq \{\vec{\mu}:Y(d,\vec{\mu},\vec{\ell}) \geq K(\vec{g_1})\}$ and $B  \triangleq \{\vec{\mu}:Y(d,\vec{\mu},\vec{\ell}) \geq K(\vec{g_2})\}$. It is clear that we have $\mathbb{P}(T^{(\vec{\ell},\vec{g}_1)} \leq d) = \mathbb{P}(A)$ and $\mathbb{P}(T^{(\vec{\ell},\vec{g}_2)} \leq d) = \mathbb{P}(B)$. Considering an arbitrary outcome of $\vec{\mu}$ with the fact $K(\vec{g_1}) \leq K(\vec{g_2})$, we have that if $Y(d,\vec{\mu},\vec{\ell}) \geq K(\vec{g_2})$ then $Y(d,\vec{\mu},\vec{\ell}) \geq K(\vec{g_1})$. It implies $B \subseteq A$ which concludes $\mathbb{P}(A) \geq \mathbb{P}(B)$, i.e., $\mathbb{P}(T^{(\vec{\ell},\vec{g_1})}(\vec{\mu}) \leq d) \geq \mathbb{P}(T^{(\vec{\ell},\vec{g_2})}(\vec{\mu}) \leq d)$. 
\section{Proof of Lemma \ref{lemma:twovalue}}\label{appendix:proof_twovalue}
Given a load allocation vector $\vec{\ell}$, we can construct $\vec{\ell^{'}}$ by assigning $\ell^{'}_i = \ell_b$ if $0 \leq \ell_i \leq \ell_b$, and  $\ell^{'}_i = \ell_g$ otherwise.

Given an outcome of $\vec{\mu}$, we denote $Y(d,\vec{\mu},\vec{\ell})$ as total number of results sent back to the master in time $d$ using the load allocation vector $\vec{\ell}$. We define two events $A \triangleq \{\vec{\mu}:Y(d,\vec{\mu},\vec{\ell}) \geq K^*\}$ and $ B \triangleq \{\vec{\mu}:Y(d,\vec{\mu},\vec{\ell^{'}}) \geq K^*\}$. It is clear that we have  $\mathbb{P}(T^{(\vec{\ell},\vec{g^*})}(\vec{\mu}) \leq d)=\mathbb{P}(A)$, and $\mathbb{P}(T^{(\vec{\ell^{'}},\vec{g^*})}(\vec{\mu}) \leq d)=\mathbb{P}(B)$. Considering an arbitrary outcome of $\vec{\mu}$, we have the following facts: (1) If $0 \leq \ell_i \leq \ell_b$, then we have $\frac{\ell^{'}_i}{\mu_i} \leq d$. (2) If $\ell_b < \ell_i \leq \ell_g$, we have either $\frac{\ell_i}{\mu_g},\frac{\ell^{'}_i}{\mu_g} \leq d$ or $\frac{\ell_i}{\mu_b},\frac{\ell^{'}_i}{\mu_b} > d$. (3) $\ell^{'}_i \geq \ell_i$ for all $i$. By the facts above, if $Y(d,\vec{\mu},\vec{\ell}) \geq K^*$, then $Y(d,\vec{\mu},\vec{\ell^{'}}) \geq K^*$ which implies $A \subseteq B$. Thus, we have $\mathbb{P}(T^{(\vec{\ell^{'}},\vec{g^{*}})}(\vec{\mu}) \leq d) \geq \mathbb{P}(T^{(\vec{\ell},\vec{g^*})}(\vec{\mu}) \leq d)$ which completes the proof. 
 \section{Proof of Lemma \ref{lemma:success}}\label{appendix:proof_success}
  In round $m$, we have the optimal success probability: $$\mathbb{P}^{*}(m) = \sum^{|\mathcal{G}^{*}_g(m)|}_{l=a(\mathcal{G}^{*}_g(m))}\sum_{\mathcal{G}:\mathcal{G}\subseteq\mathcal{G}^{*}_g(m),|\mathcal{G}|=l} \prod_{i \in \mathcal{G}} p_{g,i}(m) \prod_{i \in \mathcal{G}^{*}_g(m) \backslash \mathcal{G}} p_{b,i}(m)$$ where $\mathcal{G}^{*}_g(m)$ characterizes the optimal load allocation vector in round $m$. Let's recall that we have $i^*_m$ to determine load allocation vector in round $m$ using LEA, i.e., $\ell_{m,i} = \ell_g \ \text{if} \ 1 \leq i \leq i^*_m$, $\ell_{m,i} = \ell_b$ otherwise. It is clear that this allocation vector is characterized by a set $\hat{\mathcal{G}}(m) = [i^*_m]$. Also, we have $w(i^*_m)=a(\hat{\mathcal{G}}(m))$ where $w(\tilde{i}) \triangleq \lceil \frac{K^* - (n-\tilde{i})\ell_b}{\ell_g} \rceil$. Thus, $\mathbb{P}_{\text{LEA}}(m)$ can be written as follows:
\begin{align*}
    &\mathbb{P}_{\text{LEA}}(m) = 
    \sum^{i^*_m}_{l=w(i^*_m)}\sum_{\mathcal{G}:\mathcal{G}\subseteq [i^*_m],|\mathcal{G}|=l} \prod_{i \in \mathcal{G}} p_{g,i}(m) \prod_{i \in [\tilde{i}] \backslash \mathcal{G}} p_{b,i}(m)\\
    &=\sum^{|\hat{\mathcal{G}}_g(m)|}_{l=a(\hat{\mathcal{G}}_g(m))}\sum_{\mathcal{G}:\mathcal{G}\subseteq \hat{\mathcal{G}}_g(m),|\mathcal{G}|=l} \prod_{i \in \mathcal{G}} p_{g,i}(m) \prod_{i \in \hat{\mathcal{G}}_g(m) \backslash \mathcal{G}} p_{b,i}(m)
\end{align*}
Note that the allocation vector characterized by $\hat{\mathcal{G}}_g(m)$ maximizes the estimated success probability defined in (\ref{eq:est_success1}) and (\ref{eq:est_success2}) which is the estimated success probability based on $\hat{p}_{g,i}(m)$ and $\hat{p}_{b,i}(m)$.

By SLLN, we have that $\hat{p}_{g,i}(m)$ converges to $p_{g,i}(m)$ and $\hat{p}_{b,i}(m)$ converges to $p_{b,i}(m)$ almost surely, as $m$ goes to infinity. For all $\epsilon >0$, there exists $m(\epsilon)$ such that $|p_{g,i}(m) -\hat{p}_{g,i}(m)|< \epsilon$ and $|p_{b,i}(m) - \hat{p}_{b,i}(m)|< \epsilon$ for all $m > m(\epsilon)$. Since $\hat{\mathcal{G}}_g(m)$ maximizes the estimated success probability based on $\hat{p}_{g,i}(m)$ and $\hat{p}_{b,i}(m)$, for all $m > m(\epsilon)$, we have
\begin{align*}
    & \mathbb{P}^{*}(m) \nonumber \\ \leq & \sum^{|\mathcal{G}^{*}_g(m)|}_{l=a(|\mathcal{G}^{*}_g(m)|)}\sum_{\mathcal{G}:\mathcal{G}\subseteq\mathcal{G}^{*}_g(m),|\mathcal{G}|=l} \prod_{i \in \mathcal{G}} (\hat{p}_{g,i}(m)+\epsilon) \prod_{i \in \mathcal{G}^{*}_g(m) \backslash \mathcal{G}} (\hat{p}_{b,i}(m)+\epsilon)\\
    = &\sum^{|\mathcal{G}^{*}_g(m)|}_{l=a(|\mathcal{G}^{*}_g(m)|)}\sum_{\mathcal{G}:\mathcal{G}\subseteq\mathcal{G}^{*}_g(m),|\mathcal{G}|=l} \prod_{i \in \mathcal{G}} \hat{p}_{g,i}(m) \prod_{i \in \mathcal{G}^{*}_g(m) \backslash \mathcal{G}} \hat{p}_{b,i}(m) + f(\epsilon)\\  
    \leq &\sum^{|\hat{\mathcal{G}}_g(m)|}_{l=a(|\hat{\mathcal{G}}_g(m)|)}\sum_{\mathcal{G}:\mathcal{G}\subseteq \hat{\mathcal{G}}_g(m),|\mathcal{G}|=l} \prod_{i \in \mathcal{G}} \hat{p}_{g,i}(m) \prod_{i \in \hat{\mathcal{G}}_g(m) \backslash \mathcal{G}} \hat{p}_{b,i}(m) + f(\epsilon)\\
     \leq & \sum^{|\hat{\mathcal{G}}_g(m)|}_{l=a(|\hat{\mathcal{G}}_g(m)|)}\sum_{\mathcal{G}:\mathcal{G}\subseteq \hat{\mathcal{G}}_g(m),|\mathcal{G}|=l} \prod_{i \in \mathcal{G}} (p_{g,i}(m)+\epsilon) \prod_{i \in \hat{\mathcal{G}}_g(m) \backslash \mathcal{G}} (p_{b,i}(m) + \epsilon)\\
    & + f(\epsilon) =  \mathbb{P}_{\text{LEA}}(m) + g(\epsilon)+f(\epsilon).
\end{align*}
 Note that $h(\epsilon) \triangleq g(\epsilon)+f(\epsilon)$ is a polynomial function of $\epsilon$ and $h(0)=0$, which implies $h(\epsilon) \rightarrow 0$ as $\epsilon \rightarrow 0$. Moreover, it is clear that $\mathbb{P}_{\text{LEA}}(m) \leq \mathbb{P}^{*}(m)$ since $\mathbb{P}^{*}(m)$ is optimal. Therefore, we can conclude that for all $\epsilon_1 > 0$, there exists $m(\epsilon_1)$ such that $|\mathbb{P}_{\text{LEA}}(m)-\mathbb{P}^{*}(m)| < \epsilon_1$ for all $m>m(\epsilon_1)$ which completes the proof.
\end{document}